\documentclass[letterpaper,11pt]{article}
\def\anon{0}

\usepackage[utf8]{inputenc}
\usepackage[table,xcdraw]{xcolor}

\usepackage{tikz}
\usetikzlibrary{calc}
\usepackage{amsmath, amsthm, amssymb, thm-restate}
\usepackage{algorithmicx}
\usepackage[ruled,vlined,linesnumbered]{algorithm2e}
\usepackage{setspace}
\usepackage{mathtools}
\usepackage[numbers]{natbib}
\usepackage{comment} 
\usepackage[most]{tcolorbox} 
\usepackage{xfrac}
\usepackage{hyperref}
\usepackage{multirow}
\usepackage{caption}
\usepackage{bm}
\usepackage{float}
\usepackage{newfloat}
\usepackage{enumitem}
\usepackage{dblfloatfix} 
\usepackage{wrapfig}
\usepackage{tikz}
\usepackage{csquotes}
\usetikzlibrary{arrows, automata}
\usetikzlibrary{trees, shapes, decorations.pathreplacing, angles, quotes}
\usetikzlibrary{shapes.geometric, positioning, arrows.meta}

\usepackage{fancyhdr}

\SetCommentSty{mycommfont}
\let\oldnl\nl
\newcommand{\nonl}{\renewcommand{\nl}{\let\nl\oldnl}}

\usepackage[margin=1in]{geometry}
\usetikzlibrary{fit}
\allowdisplaybreaks

\definecolor{mygreen}{RGB}{10,70,250}
\definecolor{myred}{RGB}{80,150,20}

\hypersetup{
     colorlinks=true,
     citecolor= mygreen,
     linkcolor= myred,
     urlcolor=black
}

\renewcommand{\epsilon}{\varepsilon}

\ifnum\anon=0
\newcommand{\sbcomment}[1]{{\textcolor{blue}{[\textbf{Soheil:} #1]}}}
\newcommand{\agcomment}[1]{{\textcolor{orange}{[\textbf{Alma:} #1]}}}
\newcommand{\mscomment}[1]{{\textcolor{red}{[\textbf{Madhu:} #1]}}}
\newcommand{\amcomment}[1]{{\textcolor{purple}{[\textbf{Amir:} #1]}}}
\else
\newcommand{\sbcomment}[1]{}
\newcommand{\agcomment}[1]{}
\newcommand{\mscomment}[1]{}
\newcommand{\amcomment}[1]{}
\newcommand{\vccomment}[1]{}
\fi

\newcommand{\mE}{\mathcal{E}}
\newcommand{\mC}{\mathcal{C}}
\newcommand{\mT}{\mathcal{T}}

\newcommand{\mB}{\mathcal{B}}
\newcommand{\mG}{\mathcal{G}}

\newcommand{\hiddencomment}[1]{}

\DeclarePairedDelimiter\card{\lvert}{\rvert}
\newcommand{\zero}{{(0)}}
\newcommand{\one}{{(1)}}
\newcommand{\bd}{\ensuremath{\mathbf{d}}}
\newcommand{\dtv}{d_{\textnormal{TV}}}
\newcommand{\ER}{Erdős–Rényi}

\newcommand{\poly}{\mathrm{poly}}

\newcommand{\Ot}{\ensuremath{\widetilde{O}}}

\usepackage[noabbrev,nameinlink]{cleveref}
\crefname{lemma}{Lemma}{Lemmas}
\crefname{theorem}{Theorem}{Theorems}
\crefname{property}{Property}{Properties}
\crefname{claim}{Claim}{Claims}
\crefname{definition}{Definition}{Definitions}
\crefname{observation}{Observation}{Observations}
\crefname{proposition}{Proposition}{Propositions}
\crefname{assumption}{Assumption}{Assumptions}
\crefname{line}{Line}{Lines}
\crefname{figure}{Figure}{Figures}
\creflabelformat{property}{(#1)#2#3}
\crefname{equation}{}{}
\crefname{section}{Section}{Sections}
\crefname{appendix}{Appendix}{Appendices}
\crefname{algCounter}{Algorithm}{Algorithms}
\Crefname{algCounter}{Algorithm}{Algorithms}

\newtheorem{lemma}{Lemma}[section]
\newtheorem{theorem}[lemma]{Theorem}

\newtheorem{definition}[lemma]{Definition}
\newtheorem{claim}[lemma]{Claim}

\newtheorem{observation}[lemma]{Observation}
\newtheorem{remark}[lemma]{Remark}
\newtheorem{assumption}[lemma]{Assumption}

\counterwithin{figure}{section}

\definecolor{mylightgray}{RGB}{240,240,240}

\algnewcommand{\IIf}[2]{\textbf{if} #1 \textbf{then} #2}
\algnewcommand{\EndIIf}{\unskip\ \algorithmicend\ \algorithmicif}
\algnewcommand{\IElse}[1]{\textbf{else} #1}

\newenvironment{graytbox}{
\par\addvspace{0.1cm}
\begin{tcolorbox}[width=\textwidth,
                  boxsep=5pt,
                  left=1pt,
                  right=1pt,
                  top=2pt,
                  bottom=2pt,
                  boxrule=0pt,
                  arc=2pt,
                  colback=mylightgray,
                  colframe=black,
                  ]
}{
\end{tcolorbox}
}

\newenvironment{whitetbox}{
\par\addvspace{0.1cm}
\begin{tcolorbox}[width=\textwidth,
                  boxsep=5pt,
                  left=1pt,
                  right=1pt,
                  top=2pt,
                  bottom=2pt,
                  boxrule=1pt,
                  arc=0pt,
                  colframe=black,
                  colback=white
                  ]
}{
\end{tcolorbox}
}

\newcounter{algCounter}

\let\oldlemma\lemma
\renewcommand{\lemma}{%
  \renewcommand{\emph}[1]{\textbf{##1}}
  \oldlemma}

\let\olddefinition\definition
\renewcommand{\definition}{%
  \renewcommand{\emph}[1]{\textbf{##1}}
  \olddefinition}

\usepackage{amsthm} 
\usepackage{etoolbox} 
\usepackage{comment} 

\ifnum\anon=0
\newcommand{\snote}[1]{{\color{brown} [Soheil: #1]}}
\newcommand{\mnote}[1]{{\color{red} [Madhu: #1]}}
\newcommand{\aanote}[1]{{\color{pink} [Amir: #1]}}
\newcommand{\agnote}[1]{{\color{blue} [Alma: #1]}}
\else  
\newcommand{\aanote}[1]{}
\newcommand{\mnote}[1]{}
\newcommand{\agnote}[1]{}
\newcommand{\snote}[1]{}
\fi 

\makeatletter
\renewcommand{\paragraph}{%
  \@startsection{paragraph}{4}%
  {\z@}{10pt}{-1em}%
  {\normalfont\normalsize\bfseries}%
}
\makeatother

\makeatletter
\patchcmd{\@algocf@start}
  {-1.5em}
  {0pt}
  {}{}
\makeatother

\title{Lower Bounds for Non-adaptive Local Computation Algorithms}

\if\anon1{}\else{
\author{
Amir Azarmehr\thanks{Northeastern University, Boston, Massachusetts, USA.  Supported in part by NSF CAREER Award CCF 2442812 and a Google Research Award. Emails: \texttt{\{azarmehr.a, s.behnezhad, ghafari.m\}@northeastern.edu}.} \and 
Soheil Behnezhad\footnotemark[1] \and
Alma Ghafari\footnotemark[1] \and
Madhu Sudan\thanks{School of Engineering and Applied Sciences, Harvard University, Cambridge, Massachusetts, USA. Supported in part by a Simons Investigator Award and NSF Award CCF 2152413. Email: \texttt{madhu@cs.harvard.edu}.}
}
}\fi 

\date{}

\begin{document}

\maketitle

\thispagestyle{empty}
\begin{abstract}

We study {\em non-adaptive} Local Computation Algorithms (LCA). A reduction of Parnas and Ron (TCS'07) turns any distributed algorithm into a non-adaptive LCA. Plugging known distributed algorithms, this leads to non-adaptive LCAs for constant approximations of maximum matching (MM) and minimum vertex cover (MVC) with complexity $\Delta^{O(\log \Delta / \log \log \Delta)}$, where $\Delta$ is the maximum degree of the graph. Allowing adaptivity, this bound can be significantly improved to $\poly(\Delta)$, but is such a gap necessary or are there better non-adaptive LCAs?

\smallskip\smallskip
Adaptivity as a resource has been studied extensively across various areas. Beyond this, we further motivate the study of non-adaptive LCAs by showing that even a modest improvement over the Parnas-Ron bound for the MVC problem would have major implications in the Massively Parallel Computation (MPC) setting. In particular, it would lead to faster truly sublinear space MPC algorithms for approximate MM, a major open problem of the area. Our main result is a lower bound that rules out this avenue for progress.

\smallskip\smallskip
Specifically, we prove that $\Delta^{\Omega(\log \Delta / \log \log \Delta)}$ queries are needed for any non-adaptive LCA computing a constant approximation of MM or MVC. This is the first separation between non-adaptive and adaptive LCAs, and already matches (up to constants in the exponent) the algorithm obtained by the black-box reduction of Parnas and Ron.

\smallskip \smallskip
Our proof blends techniques from two separate lines of work: sublinear time lower bounds and distributed lower bounds. Particularly, we adopt techniques such as couplings over acyclic subgraphs from the recent sublinear time lower bounds of Behnezhad, Roghani, and Rubinstein (STOC'23, FOCS'23, STOC'24). We apply these techniques on a very different instance, particularly (a modified version of) the construction of Kuhn, Moscibroda and Wattenhoffer (JACM'16) from distributed computing. Our proof reveals that the (modified) KMW instance has the rather surprising property that any random walk of {\em any length} has a tiny chance $(\Delta^{-\Omega(\log \Delta/\log \log \Delta)})$ of identifying a matching edge. In contrast, the work of KMW only proves that short walks (i.e., walks of depth $O(\log \Delta/ \log \log \Delta)$) are not useful.

\end{abstract}


\clearpage

\setcounter{page}{1}

\section{Introduction}

Local Computation Algorithms (LCAs), formalized in the works of \citet{AlonRVX12} and \citet{RubinfeldTVX11}, are an important class of sublinear time algorithms that have had deep consequences in several other models of computation. LCAs are not required to return their entire solutions at once, but rather parts of it upon queries. For example, an LCA for the minimum vertex cover (MVC) problem only returns whether a queried vertex $v$ is part of the vertex cover, without necessarily revealing any information about the other vertices. The main measure of complexity for LCAs is the time spent by the algorithm to return the solution to one such query.

Most LCAs in the literature are {\em adaptive}. Take the vertex cover example. To determine whether $v$ is in the vertex cover, existing LCAs often explore a local neighborhood of $v$ in an {\em adaptive} manner until they are ready to decide whether $v$ is part of the output. The notion of {\em non-adaptivity} for LCAs has also been studied indirectly in several applications of LCAs. 
Indeed, this dates back to a reduction of \citet{ParnasRon07} from 2007 (in fact before LCAs were formalized) which translates any local distributed algorithm to a {\em non-adaptive} LCA. However, to our knowledge, this notion has never been formalized. We give a formal definition in this work (see \cref{def:non-adaptive-lca-general}). We describe it informally below.

\paragraph{Non-adaptive LCAs:} We recall the notion of an LCA using the specific example of (approximate) MVC. Such an algorithm is given as input a vertex $u$ in the graph, and query access to the adjacency list of the graph, and is required to output whether $u$ is in the vertex cover. The requirement is that the LCA produces answers that are consistent with {\em some} (approximate) MVC. A {\em non-adaptive} LCA for a graph problem, given a vertex $u$, non-adaptively specifies the structure of a rooted tree $T$ which is then used to explore the neighborhood of $u$. For example, setting $T$ to be a path of length $\ell$ corresponds to a random walk of length $\ell$ from $u$. Another example is setting $T$ to be a $\Delta$-ary tree of depth $r$, which corresponds to collecting the whole neighborhood of $u$ up to distance $r$. For the formal definition, see \cref{def:non-adaptive-lca-general}.

Returning to the above-mentioned reduction of \citet{ParnasRon07}, their work  observed that by collecting the $r$-hop neighborhood of a vertex $u$, one can simulate the output of any distributed local algorithm on that vertex $u$. This corresponds to a non-adaptive LCA with query complexity $\Delta^{r}$. Plugging in the $r = O(\log \Delta / \log \log \Delta)$-round distributed algorithm of \citet*{Bar-YehudaCS17}, this implies, for any fixed $\epsilon > 0$, a non-adaptive LCA for $(2+\epsilon)$ approximate MVC with query complexity $\Delta^{O(\log \Delta / \log \log \Delta)}$. While this remains the best known {\em non-adaptive} LCA for any $O(1)$ approximate MVC, the query complexity can be significantly improved if we allow adaptivity. In particular, 
there are several adaptive LCAs that obtain a $(2+\epsilon)$-approximate MVC with $\poly(\Delta)$ queries \cite{YoshidaYISTOC09,behnezhad2021,Ghaffari-FOCS22,KapralovSODA20}.
But is this large gap necessary? And should we even care about designing more efficient non-adaptive LCAs?

\vspace{-0.2cm}
\subsection{Why Non-adaptive LCAs Matter}

The role of adaptivity in algorithm design has long been studied across various areas \cite{GoldreichT01,FischerLNRRS02,ChenSTW-APX17,CanonneG18,GonenR07,RaskhodnikovaS06,0001DNSW25,DBLP:conf/stoc/GirishST024,BehnezhadDH20,IndykPW11}. The distinction between adaptive and non-adaptive algorithms is particularly well studied in property testing and learning (the list is long; see \cite{GoldreichT01,FischerLNRRS02,0001DNSW25} for some representative results and the references therein). There is an inherent connection between LCAs and such algorithms. For example, in a remarkable result, \citet*{LangeRV22} used vertex cover LCAs for properly learning monotone functions. A non-adaptive LCA would make the algorithm of \cite{LangeRV22} also non-adaptive, but unfortunately the existing non-adaptive vertex cover LCAs are too slow to be applicable for this while the adaptive ones are fast enough.

\vspace{-0.2cm}
\paragraph{Implications in the MPC model:} Even if adaptivity is not a first-order concern, there are still strong reasons to study non-adaptive LCAs due to their applications. For example, we show that even a slightly more efficient non-adaptive LCA would resolve a major open problem in the study of {\em massively parallel computation} (MPC) algorithms. The MPC model is an abstraction of many modern parallel frameworks such as MapReduce, Hadoop, or Spark \cite{KarloffSV10}. For many problems, including $O(1)$ approximate maximum matching, the fastest known MPC algorithm with truly sublinear space runs in $O(\sqrt{\log n} \cdot \log \log n)$ rounds \cite{GhaffariU19,OnakArxiv18}. There are reasons to believe the right answer is $\poly(\log \log n)$ for this problem, but achieving this or any improvement over the state-of-the-art remains a major open problem.

In this paper, we prove that any $O(1)$-approximate MVC non-adaptive LCA with complexity $O(\Delta^{(\log \Delta)^{1-\epsilon}})$ would imply an $(\log n)^{1/2-\Omega(\epsilon)}$ round MPC algorithm (formalized as \cref{thm:MPC}). Therefore, any slight improvement over the state-of-the-art $\Delta^{O(\log \Delta/\log \log \Delta)}$-query non-adaptive LCA for MVC \cite{ParnasRon07,Bar-YehudaCS17} leads to an improvement in the MPC model.

In fact, if non-adaptive LCAs were as efficient as adaptive ones with $\poly(\Delta)$ (or even $\Delta^{O(\log \log \Delta)}$) complexity, then the resulting MPC algorithm of our \cref{thm:MPC} would only run in $\poly(\log \log n)$ rounds, fully resolving the aforementioned open problem.

\subsection{Our Main Contribution}

In this paper, we prove a firm separation between adaptive and non-adaptive LCAs for computing constant approximations of minimum vertex cover, maximum matching, and maximal independent set problems. 

\begin{graytbox}
\begin{theorem}\label{thm:main}
    Any (possibly randomized) non-adaptive LCA that given a graph $G$, returns an $O(1)$ approximation of maximum matching, an $O(1)$ approximation of minimum vertex cover, or a maximal independent set with constant probability, requires $\Delta^{\Omega(\log \Delta / \log \log \Delta)}$ queries.
\end{theorem}
\end{graytbox}

\cref{thm:main} shows that the $\Delta^{O(\log \Delta / \log \log \Delta)}$-query non-adaptive LCA implied by the reduction of \cite{ParnasRon07} applied on the distributed algorithm of \cite{Bar-YehudaCS17} is optimal up to constant factors in the exponent.

\cref{thm:main} also rules out the possibility of achieving a faster MPC algorithm for approximate matchings by designing better non-adaptive LCAs. Therefore to improve over the $\widetilde{O}(\sqrt{\log n})$ bound, if at all possible, one has to follow a completely different approach.

\subsection{Perspective: The Need for New Lower Bound Techniques}

Recall again from the reduction of \citet*{ParnasRon07} that an $r$-round distributed algorithm implies an $\Delta^r$-query non-adaptive LCA. This means that any $\Delta^{\Omega(r)}$ query lower bound for non-adaptive LCAs also implies an $\Omega(r)$ lower bound on the round complexity of distributed algorithms. But the reverse is not true. To our knowledge, there are only two existing methods for proving distributed lower bounds, namely, $(i)$ round elimination, and $(ii)$ isomorphism up to a certain radius. But neither of these approaches seems useful for proving non-adaptive LCA lower bounds.

The round elimination technique \cite{linial1992locality,balliu2021lower} relies crucially on the notion of rounds, and proceeds by showing that an $r$-round algorithm for some problem $P$ implies an $(r-1)$-round algorithm for a slightly easier problem $P'$, which is then applied iteratively to obtain some contradiction. But since there is no clear analog of the notion of round complexity for non-adaptive LCAs it is not immediately clear how to apply this method for proving non-adaptive LCA lower bounds.

The second method relies on the fact that an $r$-round distributed algorithm can only view the $r$-hop neighborhood of each vertex. The lower bound proofs then follow by constructing a family of graphs where two vertices that have to produce different outputs have isomorphic $r$-hop neighborhoods. For example, \citet*{KuhnMW16} use this method to prove an $\Omega(\log \Delta / \log \log \Delta)$ lower bound on the round-complexity of any distributed algorithm computing a constant approximate MVC. The downside of this method is that if we go one hop further, the lower bound completely breaks. Namely, the said vertices with indistinguishable $r$-hop neighborhoods are completely distinguishable once we go one hop further. Since a non-adaptive LCA can explore a long but thin tree around a vertex, any lower bound should bound the effect of such explorations.

\paragraph{Our new lower bound technique:} In order to prove \cref{thm:main}, we blend techniques from two separate lines of work: sublinear time lower bounds and distributed lower bounds. Particularly, we adopt techniques such as couplings over acyclic subgraphs from the recent sublinear time lower bounds of \citet*{Behnezhad-RRS-STOC23,BehnezhadRR-FOCS23}. However, we apply these techniques very differently in that both the coupling as well as the instance it is defined on are different. Particularly, the instance is a modified version of the construction of  KMW \cite{KuhnMW16}. Our proof reveals that the (modified) KMW instance has the rather surprising property that any random walk of {\em any length} has a tiny chance $(\Delta^{-\Omega(\log \Delta/\log \log \Delta)})$ of identifying a matching edge. In contrast, the work of KMW only proves that walks of depth $O(\log \Delta/ \log \log \Delta)$ are not useful. 
\section{Technical Overview} \label{sec:technical-overview}

We describe our lower bound for maximum matching approximation.
The argument needs only slight modification for the vertex cover and the maximal independent set problems.

As noted in the previous section, our starting point is a construction of \citet*{KuhnMW16} who show that any distributed LOCAL algorithm computing an approximately maximum matching requires $r = \Omega(\log \Delta / \log\log \Delta)$ rounds of communication. To get this result, they construct a family of hard instances, given by distributions of graphs on $n$ vertices of maximum degree $\Delta$ (for every $\Delta$ and every sufficiently large $n$).
On a high level, the vertices are divided into groups, referred to as clusters, and the edges are then divided into groups based on the endpoint clusters (depicted in \Cref{fig:cluster-tree}; we have suppressed some of the details here which are explained in \cref{sec:construction}).
In particular, there is a group of \enquote{significant edges} which are indistinguishable from a group of \enquote{misleading edges} for $r$-round LOCAL algorithms.
Any constant approximation of the maximum matching must contain a constant fraction of the significant edges, whereas any matching can only contain an $o(1)$-fraction of the misleading edges.
Therefore, an algorithm that cannot distinguish between the two would be unable to approximate the maximum matching.

For the construction of KMW (and often in the context of LOCAL algorithms),
indistinguishability of vertices/edges is established by showing that the corresponding $r$-hop neighborhoods are isomorphic.
The input graphs have girth $2r + 1$, hence the $r$-hop neighborhood is a tree for any vertex.
As a result, isomorphism boils down to the degrees of the explored vertices, as the structure of any tree can be expressed using its degree sequence.
Their construction is such that the degree is the same for the vertices of a cluster, and even the same among \emph{most} clusters.
This essentially hides the cluster of an explored vertex,
and hence plays a key role in the proof of isomorphism.

Our lower bound is also obtained from the indistinguishability of significant and misleading edges. However, unlike KMW, we (provably) cannot obtain a non-adaptive LCA lower bound of $\Delta^{\Omega(\log \Delta/\log \log \Delta})$ (or even a much weaker lower bound of $\omega(\log \Delta / \log \log \Delta)$) via exact radius-based isomorphism. This is because a non-adaptive LCA of complexity  $\Delta^{o(\log \Delta / \log \log \Delta)}$ can still run a long random walk of this length, whereas the full $O(\log \Delta / \log \log \Delta)$-hop of a vertex is always sufficient to solve constant-approximate matchings due to the distributed algorithm of \cite{Bar-YehudaCS17}.

So instead of isomorphism, we blend some of the machinery developed in recent lower bounds of \citet*{Behnezhad-RRS-STOC23,BehnezhadRR-FOCS23,behnezhad2024approximating} for sublinear time lower bounds with the construction of KMW. In particular, we prove that any non-adaptive structure of size $\Delta^{o(\log \Delta/\log \log \Delta)}$ will likely see the same output when the origin is either a significant or a misleading edge. We do this by presenting a careful coupling of these explored subgraphs. This coupling implies that it is highly unlikely (but not impossible, as in the isomorphism-based proof of \cite{KuhnMW16}) that the algorithm successfully distinguishes significant edges from misleading ones. We note that while both our work and those of \cite{Behnezhad-RRS-STOC23,BehnezhadRR-FOCS23,behnezhad2024approximating} rely on couplings, the details are significantly different as the underlying graph distributions are completely different.

Our first step to formalize this coupling is to show (through a different method than bounding the girth) that the explored subgraph is a tree with high probability. The proof of this part is similar to \cite{Behnezhad-RRS-STOC23,BehnezhadRR-FOCS23,behnezhad2024approximating}. Having this structure, the degree of the vertices in this tree is the only observation the algorithm can make.
We present a coupling such that, with constant probability, the degrees of the explored vertices from the two starting points are exactly the same. This is where our proof completely deviates from the lower bounds of \cite{Behnezhad-RRS-STOC23,BehnezhadRR-FOCS23,behnezhad2024approximating}.
As a result, we can bound the probability that the algorithm can distinguish between the two types of edges.
Note that, for a non-adaptive LCA, the distributions of the observed degrees are determined by the structure of the query tree.
For a specific query tree, we loosely refer to the total variation distance of the two distributions as the amount of information revealed by the query tree. 

Take $r = \Theta(\log \Delta/\log \log \Delta)$ to be the radius in the construction of KMW up to which exact isomorphism holds between the neighborhoods of misleading and significant edges. A key challenge in formalizing our lower bound is that a non-adaptive LCA with query-complexity $\Delta^{o(r)}$ can still explore a path of length $\Delta^{o(r)}$, which goes far beyond radius $r$.
On a high level, we first address extremely long query paths (of length $\omega(r \log \Delta)$) by modifying the construction of KMW.
Then, we present a more sophisticated analysis of the KMW construction to argue about moderately long paths (of length between $r$ and $O(r \log \Delta)$). In what follows, we provide more details about each of these two steps.

\paragraph{Step 1: Protecting from extremely long walks by modifying KMW.}

As a first step, we modify the construction so that extremely long walks (of length $\omega(r \log \Delta)$) are not useful in identifying the origin.
To do this, we add a cluster of ``dummy'' vertices with edges to random vertices of the KMW graph. This cluster includes a small number ($< \epsilon n$) of vertices, so doesn't change the size of maximum matching much, but a constant fraction of neighbors of each vertex in KMW go to this cluster. This addition gives the property that every step of a random walk enters the dummy cluster with a constant probability.
Once it does, further exploration from these vertices does not provide any distinguishing information.
That is, considering the coupling of two paths starting from a significant edge and a misleading edge, once the paths reach the dummy cluster, they can be coupled perfectly, meaning the degrees of the explored vertices from that point on will be the same with probability $1$.

As a result, paths of length $\omega(r \log \Delta)$ enter the the dummy cluster at some point with probability $1-2^{-\omega(r \log \Delta)} = 1-o(\Delta^{-r})$.
This implies that the endpoint of this path reveals no useful information about the origin.
Taking the union bound over all $O(\Delta^r)$ such endpoints, we conclude that exploring a random walk further than $\omega(r \log \Delta)$ steps is not helpful to the algorithm.

\paragraph{Step 2 (main challenge): Protecting from moderately long walks.}
While the addition of dummy vertices ensures that long random walks provide no useful information, it is still possible that a ``moderately long'' random walk of length $O(r \log \Delta)$ will reveal useful information about the origin with noticeably high probability. 
Note, this is still much higher than radius $r$ for which the isomorphism guarantee holds.
As we wish to prove a lower bound of roughly $\Delta^{r}$ on the number of queries, if moderate random walks distinguish the origin with probability $\Delta^{-o(r)}$, this would preclude the kind of query lower bound we seek. 

To reason about the inability of moderate random walks to identify the origin, we need to perform a more refined analysis of the KMW construction
and to characterize the paths that serve to identify the origin.
In particular, we need to show that most of the moderate paths do not identify the origin. 
That is, considering the coupling of the two explored subgraphs starting from a significant or a misleading edge, the probability that the moderate walks have a different sequence of degrees is small.
To this end, at any point, we quantify how close the path is to revealing useful information.
Based on this distance, we classify a group of \enquote{low-labelled edges}, traversing which would bring the path closer to identifying the origin. 
We bound the number of low-labelled edges in every step, in turn bounding the probability that the path traverses one.
Therefore, we can conclude that a \enquote{distinguishing path} (i.e.\ a path that reveals the origin, \Cref{fig:distinguishing-1,fig:distinguishing-2}), which contains many such steps, has a small probability of being explored.
Finally, we remark that while the probabilistic bounds we seek happen to hold in the KMW construction, this aspect seems to be a coincidence. A priori, there seems to be no reason why the paths that lead to the identifiability of the origin should have a low probability in an instance designed to foil distributed LOCAL algorithms!

\paragraph{Additional modifications (cycles).} 
Proving the algorithm does not uncover a cycle is essential to our analysis.
The simple structure of a tree makes it possible to easily examine the distribution of the explored subgraph, as the structure of any tree can be expressed by its degree sequence.
To prove the explored subgraph is a tree, we employ techniques from sublinear lower bounds \cite{BehnezhadRR-FOCS23}:
We ensure the graph is relatively sparse, i.e.\ $\Delta = O(2^{\sqrt{\log n \log \log n}})$, and make use of random regular bipartite subgraphs.
We show that while exploring the graph, conditioned on the edges revealed so far, the remainder of each bipartite subgraph acts similarly to an \ER{} graph of the same density.
Using this lemma, we are able to bound the probability that certain edges exist (namely, for unexplored edges between explored vertices, the existence of which would imply a cycle), and argue that the algorithm does not discover any cycles.
In contrast, KMW construction uses lifting to ensure a girth of $2r + 1$, thereby preventing any $r$-round LOCAL algorithm from discovering a cycle.
Note that such a scheme is not applicable to LCAs since the algorithm can explore the graph up to radius $\Delta^r$, as guaranteeing such a high girth is infeasible.
Our construction could potentially include cycles of length $4$.
Despite that, we show that the probability of a random walk discovering a cycle is negligible.

\paragraph{Putting things together: indistinguishability through coupling.}
To summarize, we prove the indistinguishability of significant and misleading edges by coupling the explored subgraph starting at each type of edge.
We ensure that the explored subgraph is a tree by incorporating random bipartite graphs into the construction.
As a result, what the algorithm observes boils down to the degrees of the explored vertices.
We present a coupling of the two explored subgraphs that maximizes the probability that the explored vertices in every step have the same degree, which results in the two sequences being the same with constant probability.
We argue about this probability for the deeper vertices of the query tree (deeper than $\omega(r \log \Delta)$), by showing that the path leading to them reaches the dummy cluster with high probability,
and for all other vertices (of depth $O(r \log \Delta)$), by characterizing paths that lead to vertices of different degrees.
From the coupling, we can infer that the algorithm observes identical subgraphs starting from a significant or a misleading edge with constant probability, and hence cannot distinguish between them.

\paragraph{Paper Organization.}
We present the notations along with the formal definition of non-adaptive LCAs in \Cref{sec:prelim}.
Next, we describe the construction and study its properties in \Cref{sec:construction}.
This section also includes the key definition of \emph{distinguishing label sequences} (\Cref{def:distinguishing-sequence}), which characterizes the random walks that reveal information about the origin.
\Cref{sec:matching-lb}, contains our main result, the non-adaptive LCA lower bound for maximum matching.
In the subsequent sections, we show that the arguments can easily be adapted to the vertex cover problem (\cref{sec:vc}) and the maximal independent set problem (\cref{sec:mis}).
Finally, in \Cref{sec:mpc}, we discuss the application of non-adaptive LCAs to the MPC model.
\section{Preliminaries}
\label{sec:prelim}

For a graph $G = (V, E)$, we use $\mu(G)$ to denote the size of the maximum matching,
that is, the size of the largest edge set $M \subseteq E$ such that no two edges in $M$ share an endpoint.
A vertex set $A$ is a \emph{vertex cover} of $G$, if any edge $e \in E$ has an endpoint in $A$,
i.e.\ the induced subgraph $G[V \setminus A]$ has no edges.
It can be seen that any vertex cover must have at least $\mu(G)$ vertices since for every edge of the maximum matching, at least one of the endpoints must be in the vertex cover.

We start by specifying the restrictions on a non-adaptive LCA before turning to its performance guarantees (given in \cref{def:NLCA-MM}). 

\begin{definition}[Non-adaptive LCA] \label{def:non-adaptive-lca-general}
    A non-adaptive LCA, given the knowledge of the maximum degree $\Delta$, upon being prompted with a vertex $u$, produces a sequence of query \linebreak instructions $(a_1, b_1), (a_2, b_2), \ldots (a_{Q}, b_{Q})$ such that $a_i \leq i$ and $1 \leq b_i \leq \Delta$. 
    Here, $Q$ is the query complexity of the algorithm. 

    Afterwards, the instructions are used to explore the graph as follows:
    $u_1 = u$ is the first discovered vertex.
    Then, for every step $1 \leq i \leq Q$,
    the newly discovered vertex $u_{i+1}$ is set to the $b_i$-th vertex in the adjacency list of $u_{a_i}$.
    If $u_{a_i}$ has fewer than $b_i$ neighbors, then $u_{i+1}$ is set to a predesignated null value $\bot$, and any further queries made to $u_{i+1}$ also return $\bot$.
    The LCA then computes the output based on the explored subgraph.
\end{definition}

For our application, we assume that after all the vertices are discovered, their \linebreak degrees $d_1, d_2, \ldots, d_{Q+1}$ is revealed to the algorithm.
This is without loss of generality since it only makes the algorithm stronger.
Furthermore, the degrees can be computed with an extra $O(\Delta)$ factor,
and thus the query complexity in the two models differs by at most an $O(\Delta)$ factor.

\begin{remark}
    Through a slight modification of our arguments (in \cref{clm:path-coupling}), the lower bound is also applicable to the model where the query instruction consists only of integers $(a_1, a_2, \ldots, a_Q)$, and $u_{i+1}$ is set to a random neighbor of $u_{a_i}$ (with or without replacement).
\end{remark}

\begin{definition}[Non-adaptive LCAs for maximum matching]\label{def:NLCA-MM}
    Given a graph $G$, a non-adaptive LCA for maximum matching is prompted to decide whether an edge $(u, v)$ is in the output.
    Hence, we slightly modify \cref{def:non-adaptive-lca-general} to initiate the sequence of discovered vertices with $u_1 = u$ and $u_2 = v$, after which the query instructions are used to explore the graph as normal.
    We say that the LCA $c$-approximates the maximum matching, if (1) the output is always a matching, and (2) it has size larger than $c \cdot \mu(G)$ with constant probability.
\end{definition}
\section{The Construction} \label{sec:construction}

This section describes the hard instance. 
First, we present an informal overview of the construction, which is based on a blueprint. 
Then, we formally define \emph{cluster trees} the main component of the blueprint (\cref{subsec:cluster-trees}).
Afterwards, we characterize \emph{distinguishing sequences}, which are walks on the tree that reveal information about the starting point (\cref{subsec:distinguishing-sequences}).
Distinguishing sequences have a key role in our lower bound approach, in showing that the presented coupling works.
Finally, we present the full blueprint and describe how the final graph is derived from it (\cref{subsec:full-blueprint,subsec:final-construction}).

The blueprint is a graph where each node represents a cluster of vertices in the final graph, and each edge represents a set of edges between the vertices of the two clusters (i.e.\ a bipartite subgraph).
This blueprint consists of (1) two isomorphic trees, (2) a dummy cluster that is connected to all the other clusters, and (3) a perfect matching between the two trees that connects corresponding clusters.
We emphasize that, compared to \cite{KuhnMW16}, the introduction of dummy clusters is specific to our construction, and we employ a different procedure for deriving the final graph from the blueprint. 

Each of the two trees in the blueprint, called the \emph{cluster trees}, has a root cluster $C_0$.
Each cluster will be assigned a size in \cref{subsec:final-construction}, which is the number of corresponding vertices in the final graph.
The root cluster contains a constant fraction of all the vertices in the graph, and the cluster sizes shrink by a factor of $\delta$ as the depth of the cluster increases.
This ensures that for any approximate maximum matching, a constant fraction of the edges come from the perfect matching between the vertices of the two root clusters.
We call these edges \emph{significant} in the proof of the lower bound, and call certain other edges \emph{misleading}.
We use a coupling to argue that the LCA is unable to distinguish between them and conclude that the approximation suffers.

\subsection{Background on KMW: Cluster Trees} \label{subsec:cluster-trees}

\newbool{showstuff}
\setbool{showstuff}{false} 

In this section, we provide some background on \emph{cluster trees} (due to \cite{KuhnMW16}), one of the main components of the blueprint.
We state some of their properties, most of which are present in their paper in some capacity. The proofs are deferred to \cref{apx:cluster-trees}.

The tree is defined recursively by a parameter $r$, and denoted with $\mT_r = (\mC_r, \mE_r)$. Here, $\mC_r$ is the set of clusters and $\mE_r$ is the set of edges.
For two adjacent clusters $X$ and $Y$, the edge $(X, Y) \in \mE_r$ is labeled with an integer $d(X, Y)$ which represents the degree of the vertices in $X$, from the edges going to $Y$ in the final construction ($d(Y, X)$ is defined similarly). We refer to $d(X, Y)$ as the degree of $X$ to $Y$.
The labels depend on a parameter $\delta> 0$ that is defined later.

\begin{definition}[The cluster tree \cite{KuhnMW16}]
    \label{def:cluster-tree}
    The base case is $r = 0$: \footnote{We use a slightly different base case than \cite{KuhnMW16}, which ultimately leads to the same tree structure.}
    \begin{align*}
        \mC_0 &:= \{C_0, C_1\}, \\
        \mE_0 &:= \{(C_0, C_1)\}, \\
        d(C_0, C_1) &:= 1, \qquad d(C_1, C_0) := \delta,
    \end{align*}
    and the tree is rooted at $C_0$.
    
    For $r > 0$, the tree $\mT_r$ is constructed by adding some leaf clusters to $\mT_{r-1}$:
    \begin{enumerate}
        \item For every non-leaf cluster $C \in \mC_{r-1}$, a child cluster $C'$ is created with $d(C, C') = \delta^r$.
        \item For every leaf cluster $C \in \mC_{r-1}$, let $p_C$ be the parent of $C$ and $i^*$ be the integer such that $d(C, p_C) = \delta^{i^*}$.
        Then, for every $i \in \{0, 1, \ldots, r\} - \{i^*\}$
        a child cluster $C'$ is added with $d(C, C') = \delta^i$.
    \end{enumerate}
    For every new edge $(C, C')$ where $C'$ is the leaf, we let the upward label $d(C', C) = \delta \cdot d(C, C')$.
    For every cluster $C$, we use $d_p(C)$ as a shorthand for $d(C, p_C)$.
\end{definition}

The \emph{degree} and \emph{color} of the clusters are defined as follows.

\begin{definition}[Cluster degrees]
    For a cluster $C \in \mC_r$, we define its degree in the tree as 
    $$
    d_\mT(C) := \sum_{C': (C, C') \in \mE_r} d(C, C'),
    $$
    We also use $\Delta_r$ to denote the maximum degree among the clusters.
\end{definition}

\begin{definition}[Color] \label{def:color}
    For a cluster $C \in \mC_r$, its color $\tau(C)$ is the integer $0 \leq i \leq r$ such that $C$ was first added in $\mT_i$.
\end{definition}

\begin{figure}
    \centering
    \resizebox{\textwidth}{!}{
\begin{tikzpicture}[>=stealth, transform shape]

\usetikzlibrary{calc}
\node[draw, rounded corners=5pt, fill=black!60, minimum width=21.5cm, minimum height=1.5cm, font=\LARGE] (Bottom) at (2.5,0) {$C_0$};

\node[draw, rounded corners=5pt, fill=white, minimum width=3.5cm, minimum height=1cm] (RectLeft) at (-5,3) {};
\node[draw, rounded corners=5pt, fill=black!20, minimum width=3.5cm, minimum height=1cm] (RectMid1) at (0,3) {};
\node[draw, rounded corners=5pt, fill=black!40, minimum width=3.5cm, minimum height=1cm] (RectMid2) at (5,3) {};
\node[draw, rounded corners=5pt, fill=black!60, minimum width=3.5cm, minimum height=1cm, font=\LARGE] (RectMid3) at (10,3) {$C_1$};

\draw[->] ([xshift=-7.5cm]Bottom.north) -- node[left]{\(\delta_3\)} (RectLeft.south);
\draw[->] ([xshift=-2.5cm]Bottom.north) -- node[left]{\(\delta_2\)} (RectMid1.south);
\draw[->] ([xshift=+2.5cm]Bottom.north) -- node[left]{\(\delta_1\)} (RectMid2.south);
\draw[->] ([xshift=+7.5cm]Bottom.north) -- node[left]{\(\delta_0\)} (RectMid3.south);

\node[draw, rounded corners=5pt, fill=white, minimum size=1cm] (M11) at (-2,6) {};
\node[draw, rounded corners=5pt, fill=white, minimum size=1cm] (M12) at (0,6) {};
\node[draw, rounded corners=5pt, fill=white, minimum size=1cm] (M13) at (2,6) {};
\draw[->] (RectMid1.north) -- node[right]{\(\delta_2\)} (M11.south);
\draw[->] (RectMid1.north) -- node[right]{\(\delta_1\)} (M12.south);
\draw[->] (RectMid1.north) -- node[right]{\(\delta_0\)} (M13.south);

\node[draw, rounded corners=5pt, fill=white, minimum size=1cm] (M21) at (3.5,6) {};
\node[draw, rounded corners=5pt, fill=black!20, minimum size=1cm] (M22) at (5,6) {};
\node[draw, rounded corners=5pt, fill=black!20, minimum size=1cm] (M23) at (6.5,6) {};
\draw[->] (RectMid2.north) -- node[right]{\(\delta_3\)} (M21.south);
\draw[->] (RectMid2.north) -- node[right]{\(\delta_1\)} (M22.south);
\draw[->] (RectMid2.north) -- node[right]{\(\delta_0\)} (M23.south);

\node[draw, rounded corners=5pt, fill=white, minimum size=1cm] (M31) at (8,6) {}; 
\node[draw, rounded corners=5pt, fill=black!20, minimum size=1cm] (M32) at (10,6) {};
\node[draw, rounded corners=5pt, fill=black!40, minimum size=1cm] (M33) at (12,6) {};
\draw[->] (RectMid3.north) -- node[right]{\(\delta_3\)} (M31.south);
\draw[->] (RectMid3.north) -- node[right]{\(\delta_2\)} (M32.south);
\draw[->] (RectMid3.north) -- node[right]{\(\delta_0\)} (M33.south);

\node[draw, rounded corners=2pt, fill=white, minimum size=0.5cm] (M22_top1) at (4,9) {};
\node[draw, rounded corners=2pt, fill=white, minimum size=0.5cm] (M22_top2) at (5,9) {};
\node[draw, rounded corners=2pt, fill=white, minimum size=0.5cm] (M22_top3) at (6,9) {};
\draw[->] (M22.north) -- node[left]{\(\delta_3\)} (M22_top1.south);
\draw[->] (M22.north) -- node[above]{\(\delta_2\)} (M22_top2.south);
\draw[->] (M22.north) -- node[right]{\(\delta_0\)} (M22_top3.south);

\node[draw, rounded corners=2pt, fill=white, minimum size=0.5cm] (M23_top1) at (7,9) {};
\node[draw, rounded corners=2pt, fill=white, minimum size=0.5cm] (M23_top2) at (8,9) {};
\node[draw, rounded corners=2pt, fill=white, minimum size=0.5cm] (M23_top3) at (9,9) {};
\draw[->] (M23.north) -- node[left]{\(\delta_3\)} (M23_top1.south);
\draw[->] (M23.north) -- node[above]{\(\delta_2\)} (M23_top2.south);
\draw[->] (M23.north) -- node[right]{\(\delta_0\)} (M23_top3.south);

\node[draw, rounded corners=2pt, fill=white, minimum size=0.5cm] (M32_top1) at (10,9) {};
\node[draw, rounded corners=2pt, fill=white, minimum size=0.5cm] (M32_top2) at (11,9) {};
\node[draw, rounded corners=2pt, fill=white, minimum size=0.5cm] (M32_top3) at (12,9) {};
\draw[->] (M32.north) -- node[right]{\(\delta_2\)} (M32_top1.south);
\draw[->] (M32.north) -- node[right]{\(\delta_1\)} (M32_top2.south);
\draw[->] (M32.north) -- node[right]{\(\delta_0\)} (M32_top3.south);

\node[draw, rounded corners=2pt, fill=white, minimum size=0.5cm] (M33_1) at (13,9) {};
\node[draw, rounded corners=2pt, fill=black!20, minimum size=0.5cm] (M33_2) at (14,9) {};
\node[draw, rounded corners=2pt, fill=black!20, minimum size=0.5cm] (M33_3) at (15,9) {};
\draw[->] (M33.north) -- node[left]{\(\delta_3\)} (M33_1.south);
\draw[->] (M33.north) -- node[above]{\(\delta_2\)} (M33_2.south);
\draw[->] (M33.north) -- node[right]{\(\delta_0\)} (M33_3.south);

\node[draw, rounded corners=2pt, fill=white, minimum size=0.5cm] (M33_2_a) at (13,12) {};
\node[draw, rounded corners=2pt, fill=white, minimum size=0.5cm] (M33_2_b) at (14,12) {};
\node[draw, rounded corners=2pt, fill=white, minimum size=0.5cm] (M33_2_c) at (15,12) {};
\draw[->] (M33_2.north) -- node[right]{\(\delta_2\)} (M33_2_a.south);
\draw[->] (M33_2.north) -- node[right]{\(\delta_1\)} (M33_2_b.south);
\draw[->] (M33_2.north) -- node[right]{\(\delta_0\)} (M33_2_c.south);

\node[draw, rounded corners=2pt, fill=white, minimum size=0.5cm] (M33_3_a) at (16,12) {};
\node[draw, rounded corners=2pt, fill=white, minimum size=0.5cm] (M33_3_b) at (17,12) {};
\node[draw, rounded corners=2pt, fill=white, minimum size=0.5cm] (M33_3_c) at (18,12) {};
\draw[->] (M33_3.north) -- node[right]{\(\delta_3\)} (M33_3_a.south);
\draw[->] (M33_3.north) -- node[right]{\(\delta_2\)} (M33_3_b.south);
\draw[->] (M33_3.north) -- node[right]{\(\delta_0\)} (M33_3_c.south);

\coordinate (A) at (-9,-1);
\coordinate (B) at (18.5,12.5);



\end{tikzpicture}
    }
    \caption{An illustration of the cluster tree (\cref{def:cluster-tree}) for $r = 3$. The shades represent the cluster colors. 
    $C_0$ is the root. The edge labels are displayed only in the direction away from the root.}
    \label{fig:cluster-tree}
\end{figure}
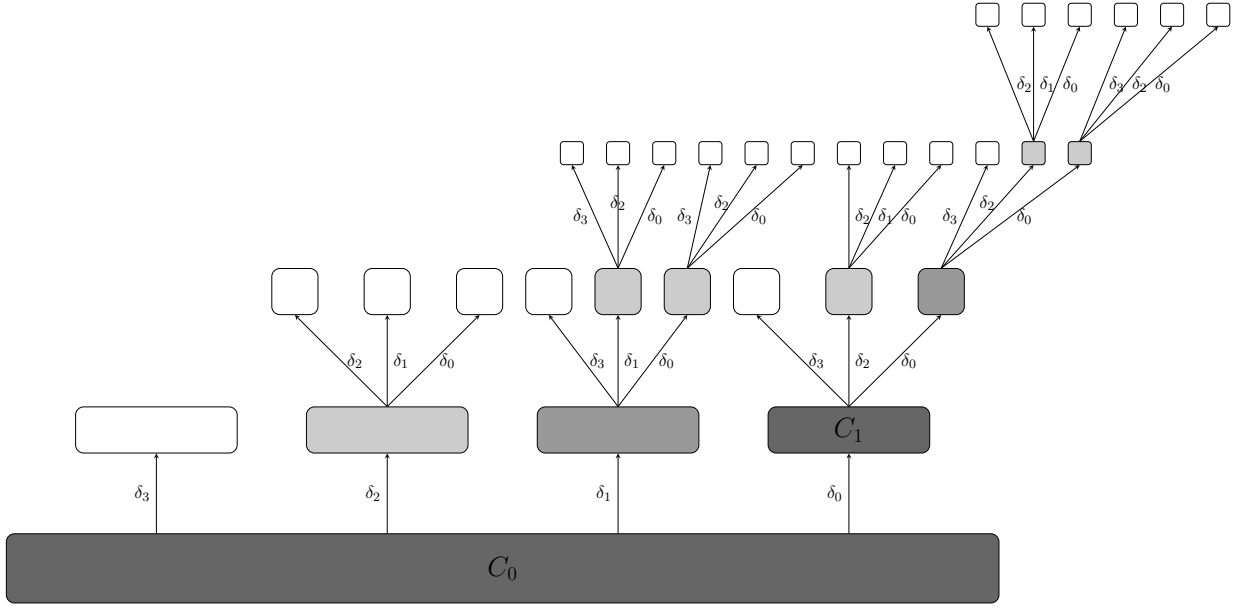

First, we characterize the degrees.

\begin{claim} \label{clm:max-parent-degree}
    For all clusters $C \in \mC_r \setminus \{C_0\}$, it holds
    $$
    d_p(C) \leq \delta^{\tau(C) + 1}.
    $$
\end{claim}
\ifbool{showstuff}{ \begin{proof}
    Recall that $C$ is added in iteration $\tau(C)$, at which point the degree to any leaf cluster is at most $\delta^{\tau(C)}$. Therefore, we have
    \begin{equation*}
    d_p(C) = d(C, p_C) = \delta \cdot d(p_C, C) \leq \delta^{\tau(C) + 1}. \qedhere
    \end{equation*}
\end{proof}}{}

\begin{claim} \label{clm:non-leaf-degrees-in-tree}
    For all non-leaf clusters $C \in \mC_r$, it holds
    $$
    d_\mT(C) = \sum_{i = 0}^r \delta^i =: \bar{d}_r.
    $$
\end{claim}
\ifbool{showstuff}{ \begin{proof}
    Take a non-leaf cluster $C$, and consider the iterative construction of $\mT_r$. Let $i^*$ be such that $d_p(C) = \delta^{i^*}$.
    Recall that $C$ is added in iteration $\tau(C) < r$ at which point it is a leaf.
    In iteration $\tau(C) + 1$, for every $i \in \{0, 1, \ldots, r\} - \{i^*\}$ a leaf $C'$ is added to $C$ with $d(C, C') = \delta^i$.
    Considering the degree to the parent $d_p(C) = \delta^{i^*}$, the total degree of $C$ at this point is
    $$
    \sum_{i=0}^{\tau(C) + 1} \delta^i.
    $$
    This proves the claim for cases where $r = \tau(C) + 1$.
    Otherwise, in the proceeding iterations $i$ from $\tau(C) + 2$ to $r$,
    a leaf $C'$ is added to $C$ with $d(C, C') = \delta^i$.
    Therefore, after the $r$-th iteration it holds:
    \begin{equation*}
    d_\mT(C) = \sum_{i = 0}^r \delta^i. \qedhere
    \end{equation*}
\end{proof}}{}

\begin{claim} \label{clm:tree-max-degree}
    The maximum degree among the clusters is $\Delta_r = \delta^{r + 1}$.
\end{claim}
\ifbool{showstuff}{ \begin{proof}
    The maximum degree among the leaf clusters is $\delta^{r+1}$,
    because for a leaf cluster $C$, it holds
    $$
    d_\mT(C) = d_p(C) = \delta \cdot d(p_C, C),
    $$
    and $d(p_C, C)$ ranges from $\delta^0$ to $\delta^r$.
    For any non-leaf cluster, by \cref{clm:non-leaf-degrees-in-tree}, the degree is 
    \begin{equation*}
    \sum_{i = 0}^r \delta^i
    = \frac{\delta^{r+1} - 1}{\delta - 1}
    \leq \delta^{r+1}.
    \qedhere
    \end{equation*}
\end{proof}}{}

The following claim states that clusters with the same total degrees
have similar outgoing labels.
This is a crucial fact in our coupling arguments.

\begin{claim} \label{clm:same-outgoing-labels-in-tree}
    Take any two clusters $C, C' \in \mC_r$.
    If $d_\mT(C) = d_\mT(C')$, then $C$ and $C'$ have the same set of outgoing edge-labels.
\end{claim}
\ifbool{showstuff}{ \begin{proof}
    If $d_\mT(C) = d_\mT(C') = \bar{d}_r$, then both $C$ and $C'$ are non-leaf clusters.
    Therefore, for every $i \in \{0, 1, \ldots, r\}$ they have exactly one edge with outgoing label $\delta^i$.
    If $d_\mT(C) = d_\mT(C') \neq \bar{d}_r$, then $C$ and $C'$ are both leaves, and they only have one edge with outgoing label $d_p(C) = d_p(C')$.
\end{proof}}{}

The following claim characterizes the color of a child cluster based on the color of the parent, and the label of the edge coming in from the parent.

\begin{claim} \label{clm:child-color}
    Take a cluster $C$ and let $C_i$ be a child of $C$ such that $d(C, C_i) = \delta^i$. Then, it holds
    $$
    \tau(C_i) = \max(i, \tau(C) + 1).
    $$
\end{claim}
\ifbool{showstuff}{ \begin{proof}
    Consider the iterative construction in \cref{def:cluster-tree}.
    Cluster $C$ is added in iteration $\tau(C)$ as a leaf.
    Let $i^*$ be such that $d_P(C) = \delta^{i^*}$.
    In iteration $\tau(C) + 1$, since $C$ is a leaf,
    for every $i \in \{0, 1, \ldots, \tau(C) + 1\} - \{i^*\}$,
    a child $C_i$ is added with $d(C, C_i) = \delta^i$.
    This proves the claim for $i \leq \tau(C) + 1$.

    In iteration $i > \tau(C) + 1$,
    since $C$ is no longer a leaf,
    a single child $C_i$ is added for $C$ with $d(C, C_i) = \delta^i$.
    This proves the claim for $i > \tau(C) + 1$.
\end{proof}}{}

The following claim states that the structure of the subtree of $C \notin \{C_0, C_1\}$ is uniquely determined by $\tau(C)$ and $d_P(C)$.
This observation is implicitly present in Definition 2 and Lemma 7 of \cite{KuhnMW16}.

\begin{claim} \label{clm:identical-subtrees}
    Let $C, C' \in \mC_r \setminus \{C_0, C_1\}$ be two clusters of the same color
    $\tau = \tau(C) = \tau(C')$ 
    that have the same degrees 
    to their parents.
    Then, the subtrees of $C$ and $C'$ are identical (considering the labels on the edges).
\end{claim}
\ifbool{showstuff}{ \begin{proof}
For a fixed $r$, we prove the claim by induction on $\tau$.
The base case is $\tau = r$ where the claim holds because both $C$ and $C'$ are leaves.

Assume the claim holds for colors larger than $\tau$,
and let $C$ and $C'$ be two clusters such that 
$$\tau(C) = \tau(C') = \tau \qquad \textnormal{and} \qquad d_p(C) = d_p(C') = \delta^{i^*}.$$
For every $i \in \{0, 1, \ldots, r\} - \{i^*\}$, cluster $C$ has a child cluster $C_i$ with $d(C, C_i) = \delta^i$. Similarly, $C'$ has a child $C'_i$ with $d(C', C'_i) = \delta^i$.
It holds $$d_p(C_i) = d_p(C'_i) = \delta^{i+1}.$$
Furthermore, by \cref{clm:child-color}, we have
$$
\tau(C_i) = \tau(C'_i) \geq \tau + 1.
$$
Therefore, by the induction hypothesis, the subtree of $C_i$ is identical to that of $C'_i$. This implies the claim for the subtrees of $C$ and $C'$ and concludes the proof.
\end{proof}}{}

The following claim uses \cref{clm:identical-subtrees}
to show that starting from two clusters $C$ and $C'$, taking the edges with the same labels leads to isomorphic subtrees if the edge-label is large enough.

\begin{claim} \label{clm:large-steps}
    Take two clusters $C, C' \in \mC_r$. Let $\tau = \max(\tau(C), \tau(C')) + 2$, and $B$ and $B'$ be \emph{neighbors} of $C$ and $C'$ resp.\ such that $d(C', B) = d(C, B') \geq \delta^{\tau}$. Then, the subtrees of $B$ and $B'$ are identical (considering the labels on the edges).
\end{claim}
\ifbool{showstuff}{ \begin{proof}
    Recall that $d_p(C) \leq \delta^{\tau(C) + 1} < \delta^\tau$ (\cref{clm:max-parent-degree}). The same holds for $C'$. 
    Therefore, $B$ and $B'$ must be children of $C$ and $C'$ respectively.
    Let $i \geq \tau$ be such that $d(C, B) = d(C', B') = \delta^i$. Then, by \cref{clm:child-color}, we have 
    $$
    \tau(B) = \max( i, \tau(C) + 1) \qquad \textnormal{and} \qquad
    \tau(B') = \max( i, \tau(C') + 1).
    $$
    Since $i \geq \tau = \max(\tau(C), \tau(C')) + 2$, it holds $\tau(B) = \tau(B')$.
    Hence, we can invoke \cref{clm:identical-subtrees} to derive the claim.
\end{proof}}{}

\setbool{showstuff}{true}

\subsection{Characterizing Distinguishing Sequences} \label{subsec:distinguishing-sequences}
Now, we move on to a key part of this section.
In \cref{sec:matching-lb}, we present a coupling which couples paths that start at different clusters and follow edges with the same labels.
For the coupling to succeed, it is critical that the corresponding clusters between the two paths have the same degrees.
Here, we examine what sequences of edge labels the paths can take
that lead to clusters with different degrees,
i.e.\ which sequences cause the coupling to fail.

\begin{definition}[Label Sequences] \label{def:label-sequence}
    Let $d_1, d_2, \ldots, d_L$ be a sequence of edge labels, i.e.\ for $1 \leq i \leq L$, we have $d_i = \delta^j$ for some integer $0 \leq j \leq r + 1$.
    The path corresponding to $\{d_i\}_i$ starting at cluster $C$ is defined as $B_0, B_1, \ldots, B_{L'}$ such that (1) $B_0 := C$, (2) $B_i$ is the neighbor of $B_{i-1}$ with $d(B_{i-1}, B_i) = d_i$, and (3) if at any point $B_{i-1}$ does not have an edge with outgoing label $d_i$, the path stops.
\end{definition}

\begin{definition}[Distinguishing Sequences] \label{def:distinguishing-sequence}
    A label sequence $d_1, d_2, \ldots, d_L$ distinguishes between $C, C' \in \mC_r$ if the corresponding paths starting at $C$ and $C'$ reach clusters $B_i$ and $B_i'$ at step $i$ such that $d_\mT(B_i) \neq d_\mT(B_i')$.
\end{definition}
\begin{remark}
    Recall that all the non-leaf clusters have the same degree (\cref{clm:non-leaf-degrees-in-tree}).
    Furthermore, when two clusters have the same degree, they have the same set of outgoing edge-labels (\cref{clm:same-outgoing-labels-in-tree}).
    Therefore, for a label sequence to distinguish between $C$ and $C'$,
    at least one of the paths must reach a leaf cluster,
    while the other one reaches a non-leaf cluster or a leaf with a different degree.
\end{remark}

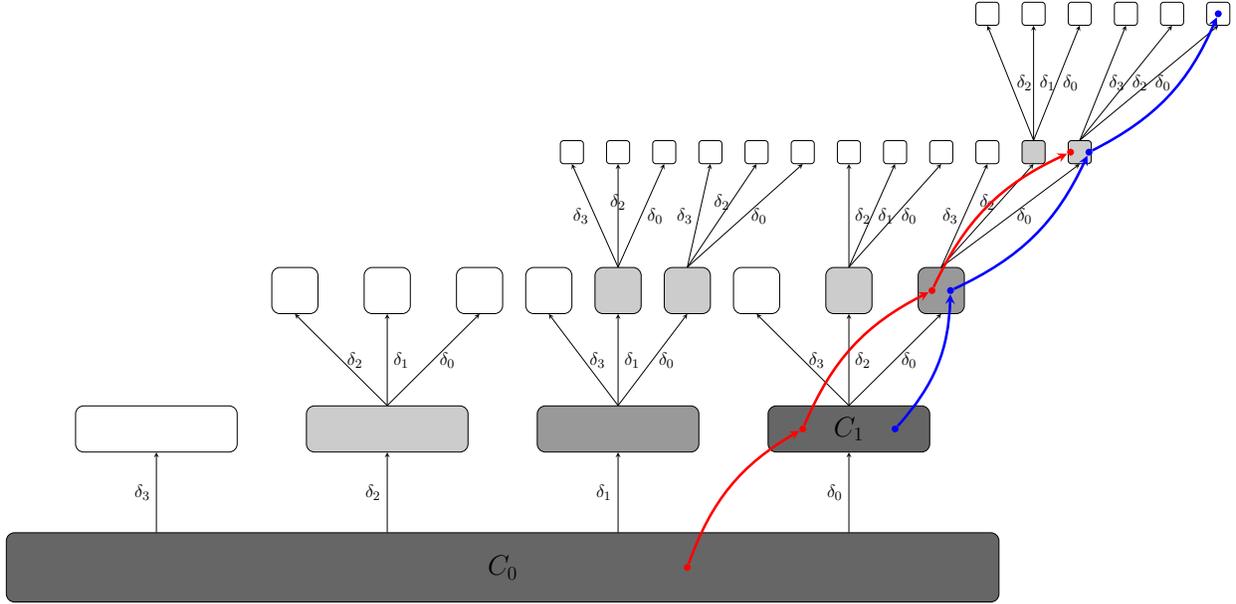
\begin{figure}
    \centering
    \resizebox{\textwidth}{!}{
\begin{tikzpicture}[>=stealth, transform shape]
\usetikzlibrary{calc}
\node[draw, rounded corners=5pt, fill=black!60, minimum width=21.5cm, minimum height=1.5cm, font=\LARGE] (Bottom) at (2.5,0) {$C_0$};

\node[draw, rounded corners=5pt, fill=white, minimum width=3.5cm, minimum height=1cm] (RectLeft) at (-5,3) {};
\node[draw, rounded corners=5pt, fill=black!20, minimum width=3.5cm, minimum height=1cm] (RectMid1) at (0,3) {};
\node[draw, rounded corners=5pt, fill=black!40, minimum width=3.5cm, minimum height=1cm] (RectMid2) at (5,3) {};
\node[draw, rounded corners=5pt, fill=black!60, minimum width=3.5cm, minimum height=1cm, font=\LARGE] (RectMid3) at (10,3) {$C_1$};

\draw[->] ([xshift=-7.5cm]Bottom.north) -- node[left]{\(\delta_3\)} (RectLeft.south);
\draw[->] ([xshift=-2.5cm]Bottom.north) -- node[left]{\(\delta_2\)} (RectMid1.south);
\draw[->] ([xshift=+2.5cm]Bottom.north) -- node[left]{\(\delta_1\)} (RectMid2.south);
\draw[->] ([xshift=+7.5cm]Bottom.north) -- node[left]{\(\delta_0\)} (RectMid3.south);

\node[draw, rounded corners=5pt, fill=white, minimum size=1cm] (M11) at (-2,6) {};
\node[draw, rounded corners=5pt, fill=white, minimum size=1cm] (M12) at (0,6) {};
\node[draw, rounded corners=5pt, fill=white, minimum size=1cm] (M13) at (2,6) {};
\draw[->] (RectMid1.north) -- node[right]{\(\delta_2\)} (M11.south);
\draw[->] (RectMid1.north) -- node[right]{\(\delta_1\)} (M12.south);
\draw[->] (RectMid1.north) -- node[right]{\(\delta_0\)} (M13.south);

\node[draw, rounded corners=5pt, fill=white, minimum size=1cm] (M21) at (3.5,6) {};
\node[draw, rounded corners=5pt, fill=black!20, minimum size=1cm] (M22) at (5,6) {};
\node[draw, rounded corners=5pt, fill=black!20, minimum size=1cm] (M23) at (6.5,6) {};
\draw[->] (RectMid2.north) -- node[right]{\(\delta_3\)} (M21.south);
\draw[->] (RectMid2.north) -- node[right]{\(\delta_1\)} (M22.south);
\draw[->] (RectMid2.north) -- node[right]{\(\delta_0\)} (M23.south);

\node[draw, rounded corners=5pt, fill=white, minimum size=1cm] (M31) at (8,6) {}; 
\node[draw, rounded corners=5pt, fill=black!20, minimum size=1cm] (M32) at (10,6) {};
\node[draw, rounded corners=5pt, fill=black!40, minimum size=1cm] (M33) at (12,6) {};
\draw[->] (RectMid3.north) -- node[right]{\(\delta_3\)} (M31.south);
\draw[->] (RectMid3.north) -- node[right]{\(\delta_2\)} (M32.south);
\draw[->] (RectMid3.north) -- node[right]{\(\delta_0\)} (M33.south);

\node[draw, rounded corners=2pt, fill=white, minimum size=0.5cm] (M22_top1) at (4,9) {};
\node[draw, rounded corners=2pt, fill=white, minimum size=0.5cm] (M22_top2) at (5,9) {};
\node[draw, rounded corners=2pt, fill=white, minimum size=0.5cm] (M22_top3) at (6,9) {};
\draw[->] (M22.north) -- node[left]{\(\delta_3\)} (M22_top1.south);
\draw[->] (M22.north) -- node[above]{\(\delta_2\)} (M22_top2.south);
\draw[->] (M22.north) -- node[right]{\(\delta_0\)} (M22_top3.south);

\node[draw, rounded corners=2pt, fill=white, minimum size=0.5cm] (M23_top1) at (7,9) {};
\node[draw, rounded corners=2pt, fill=white, minimum size=0.5cm] (M23_top2) at (8,9) {};
\node[draw, rounded corners=2pt, fill=white, minimum size=0.5cm] (M23_top3) at (9,9) {};
\draw[->] (M23.north) -- node[left]{\(\delta_3\)} (M23_top1.south);
\draw[->] (M23.north) -- node[above]{\(\delta_2\)} (M23_top2.south);
\draw[->] (M23.north) -- node[right]{\(\delta_0\)} (M23_top3.south);

\node[draw, rounded corners=2pt, fill=white, minimum size=0.5cm] (M32_top1) at (10,9) {};
\node[draw, rounded corners=2pt, fill=white, minimum size=0.5cm] (M32_top2) at (11,9) {};
\node[draw, rounded corners=2pt, fill=white, minimum size=0.5cm] (M32_top3) at (12,9) {};
\draw[->] (M32.north) -- node[right]{\(\delta_2\)} (M32_top1.south);
\draw[->] (M32.north) -- node[right]{\(\delta_1\)} (M32_top2.south);
\draw[->] (M32.north) -- node[right]{\(\delta_0\)} (M32_top3.south);

\node[draw, rounded corners=2pt, fill=white, minimum size=0.5cm] (M33_1) at (13,9) {};
\node[draw, rounded corners=2pt, fill=black!20, minimum size=0.5cm] (M33_2) at (14,9) {};
\node[draw, rounded corners=2pt, fill=black!20, minimum size=0.5cm] (M33_3) at (15,9) {};
\draw[->] (M33.north) -- node[left]{\(\delta_3\)} (M33_1.south);
\draw[->] (M33.north) -- node[above]{\(\delta_2\)} (M33_2.south);
\draw[->] (M33.north) -- node[right]{\(\delta_0\)} (M33_3.south);

\node[draw, rounded corners=2pt, fill=white, minimum size=0.5cm] (M33_2_a) at (13,12) {};
\node[draw, rounded corners=2pt, fill=white, minimum size=0.5cm] (M33_2_b) at (14,12) {};
\node[draw, rounded corners=2pt, fill=white, minimum size=0.5cm] (M33_2_c) at (15,12) {};
\draw[->] (M33_2.north) -- node[right]{\(\delta_2\)} (M33_2_a.south);
\draw[->] (M33_2.north) -- node[right]{\(\delta_1\)} (M33_2_b.south);
\draw[->] (M33_2.north) -- node[right]{\(\delta_0\)} (M33_2_c.south);

\node[draw, rounded corners=2pt, fill=white, minimum size=0.5cm] (M33_3_a) at (16,12) {};
\node[draw, rounded corners=2pt, fill=white, minimum size=0.5cm] (M33_3_b) at (17,12) {};
\node[draw, rounded corners=2pt, fill=white, minimum size=0.5cm] (M33_3_c) at (18,12) {};
\draw[->] (M33_3.north) -- node[right]{\(\delta_3\)} (M33_3_a.south);
\draw[->] (M33_3.north) -- node[right]{\(\delta_2\)} (M33_3_b.south);
\draw[->] (M33_3.north) -- node[right]{\(\delta_0\)} (M33_3_c.south);
\node[fill=red, circle, inner sep=1.5pt] (dot_Bottom_red) at ([xshift=40mm]Bottom.center) {};
\node[fill=red, circle, inner sep=1.5pt] (dot_RMid3_red) at ([xshift=-10mm]RectMid3.center) {};
\node[fill=red, circle, inner sep=1.5pt] (dot_M33_red) at ([xshift=-2mm]M33.center) {};
\node[fill=red, circle, inner sep=1.5pt] (dot_M33_3_red) at ([xshift=-2mm]M33_3.center) {};

\node[fill=blue, circle, inner sep=1.5pt] (dot_RMid3_blue) at ([xshift=10mm]RectMid3.center) {};
\node[fill=blue, circle, inner sep=1.5pt] (dot_M33_blue) at ([xshift=2mm]M33.center) {};
\node[fill=blue, circle, inner sep=1.5pt] (dot_M33_3_blue) at ([xshift=2mm]M33_3.center) {};
\node[fill=blue, circle, inner sep=1.5pt] (dot_M33_3_c_blue) at (M33_3_c.center) {};

\draw[->,red, ultra thick] (dot_Bottom_red) to[bend left=20] (dot_RMid3_red);

\draw[->,red, ultra thick]
(dot_RMid3_red)
to[bend left=20] (dot_M33_red);

\draw[->,red, ultra thick]
(dot_M33_red)to[bend left=20] (dot_M33_3_red);

\draw[->,blue, ultra thick] (dot_RMid3_blue) to[bend right=20] (dot_M33_blue); 

\draw[->,blue, ultra thick]
(dot_M33_blue)
to[bend right=20] (dot_M33_3_blue) ;

\draw[->,blue, ultra thick]
(dot_M33_3_blue)
to[bend right=20] (dot_M33_3_c_blue);




\end{tikzpicture}
    }
    \caption{An example of distinguishing sequences, following label sequence $(\delta_0, \delta_0, \delta_0)$.
    The paths start at $C_0$ and $C_1$. 
    One reaches a non-leaf cluster that has degree $\delta_0 + \delta_1 + \delta_2 + \delta_3$,
    whereas the other one reaches a leaf of degree $\delta_1$.
    }
    \label{fig:distinguishing-1}
\end{figure}
\begin{figure}
    \centering
    \resizebox{\textwidth}{!}{
\begin{tikzpicture}[>=stealth, transform shape]
\usetikzlibrary{calc}
\node[draw, rounded corners=5pt, fill=black!60, minimum width=21.5cm, minimum height=1.5cm, font=\LARGE] (Bottom) at (2.5,0) {$C_0$};

\node[draw, rounded corners=5pt, fill=white, minimum width=3.5cm, minimum height=1cm] (RectLeft) at (-5,3) {};
\node[draw, rounded corners=5pt, fill=black!20, minimum width=3.5cm, minimum height=1cm] (RectMid1) at (0,3) {};
\node[draw, rounded corners=5pt, fill=black!40, minimum width=3.5cm, minimum height=1cm] (RectMid2) at (5,3) {};
\node[draw, rounded corners=5pt, fill=black!60, minimum width=3.5cm, minimum height=1cm, font=\LARGE] (RectMid3) at (10,3) {$C_1$};

\draw[->] ([xshift=-7.5cm]Bottom.north) -- node[left]{\(\delta_3\)} (RectLeft.south);
\draw[->] ([xshift=-2.5cm]Bottom.north) -- node[left]{\(\delta_2\)} (RectMid1.south);
\draw[->] ([xshift=+2.5cm]Bottom.north) -- node[left]{\(\delta_1\)} (RectMid2.south);
\draw[->] ([xshift=+7.5cm]Bottom.north) -- node[left]{\(\delta_0\)} (RectMid3.south);

\node[draw, rounded corners=5pt, fill=white, minimum size=1cm] (M11) at (-2,6) {};
\node[draw, rounded corners=5pt, fill=white, minimum size=1cm] (M12) at (0,6) {};
\node[draw, rounded corners=5pt, fill=white, minimum size=1cm] (M13) at (2,6) {};
\draw[->] (RectMid1.north) -- node[right]{\(\delta_2\)} (M11.south);
\draw[->] (RectMid1.north) -- node[right]{\(\delta_1\)} (M12.south);
\draw[->] (RectMid1.north) -- node[right]{\(\delta_0\)} (M13.south);

\node[draw, rounded corners=5pt, fill=white, minimum size=1cm] (M21) at (3.5,6) {};
\node[draw, rounded corners=5pt, fill=black!20, minimum size=1cm] (M22) at (5,6) {};
\node[draw, rounded corners=5pt, fill=black!20, minimum size=1cm] (M23) at (6.5,6) {};
\draw[->] (RectMid2.north) -- node[right]{\(\delta_3\)} (M21.south);
\draw[->] (RectMid2.north) -- node[right]{\(\delta_1\)} (M22.south);
\draw[->] (RectMid2.north) -- node[right]{\(\delta_0\)} (M23.south);

\node[draw, rounded corners=5pt, fill=white, minimum size=1cm] (M31) at (8,6) {}; 
\node[draw, rounded corners=5pt, fill=black!20, minimum size=1cm] (M32) at (10,6) {};
\node[draw, rounded corners=5pt, fill=black!40, minimum size=1cm] (M33) at (12,6) {};
\draw[->] (RectMid3.north) -- node[right]{\(\delta_3\)} (M31.south);
\draw[->] (RectMid3.north) -- node[right]{\(\delta_2\)} (M32.south);
\draw[->] (RectMid3.north) -- node[right]{\(\delta_0\)} (M33.south);

\node[draw, rounded corners=2pt, fill=white, minimum size=0.5cm] (M22_top1) at (4,9) {};
\node[draw, rounded corners=2pt, fill=white, minimum size=0.5cm] (M22_top2) at (5,9) {};
\node[draw, rounded corners=2pt, fill=white, minimum size=0.5cm] (M22_top3) at (6,9) {};
\draw[->] (M22.north) -- node[left]{\(\delta_3\)} (M22_top1.south);
\draw[->] (M22.north) -- node[above]{\(\delta_2\)} (M22_top2.south);
\draw[->] (M22.north) -- node[right]{\(\delta_0\)} (M22_top3.south);

\node[draw, rounded corners=2pt, fill=white, minimum size=0.5cm] (M23_top1) at (7,9) {};
\node[draw, rounded corners=2pt, fill=white, minimum size=0.5cm] (M23_top2) at (8,9) {};
\node[draw, rounded corners=2pt, fill=white, minimum size=0.5cm] (M23_top3) at (9,9) {};
\draw[->] (M23.north) -- node[left]{\(\delta_3\)} (M23_top1.south);
\draw[->] (M23.north) -- node[above]{\(\delta_2\)} (M23_top2.south);
\draw[->] (M23.north) -- node[right]{\(\delta_0\)} (M23_top3.south);

\node[draw, rounded corners=2pt, fill=white, minimum size=0.5cm] (M32_top1) at (10,9) {};
\node[draw, rounded corners=2pt, fill=white, minimum size=0.5cm] (M32_top2) at (11,9) {};
\node[draw, rounded corners=2pt, fill=white, minimum size=0.5cm] (M32_top3) at (12,9) {};
\draw[->] (M32.north) -- node[right]{\(\delta_2\)} (M32_top1.south);
\draw[->] (M32.north) -- node[right]{\(\delta_1\)} (M32_top2.south);
\draw[->] (M32.north) -- node[right]{\(\delta_0\)} (M32_top3.south);

\node[draw, rounded corners=2pt, fill=white, minimum size=0.5cm] (M33_1) at (13,9) {};
\node[draw, rounded corners=2pt, fill=black!20, minimum size=0.5cm] (M33_2) at (14,9) {};
\node[draw, rounded corners=2pt, fill=black!20, minimum size=0.5cm] (M33_3) at (15,9) {};
\draw[->] (M33.north) -- node[left]{\(\delta_3\)} (M33_1.south);
\draw[->] (M33.north) -- node[above]{\(\delta_2\)} (M33_2.south);
\draw[->] (M33.north) -- node[right]{\(\delta_0\)} (M33_3.south);

\node[draw, rounded corners=2pt, fill=white, minimum size=0.5cm] (M33_2_a) at (13,12) {};
\node[draw, rounded corners=2pt, fill=white, minimum size=0.5cm] (M33_2_b) at (14,12) {};
\node[draw, rounded corners=2pt, fill=white, minimum size=0.5cm] (M33_2_c) at (15,12) {};
\draw[->] (M33_2.north) -- node[right]{\(\delta_2\)} (M33_2_a.south);
\draw[->] (M33_2.north) -- node[right]{\(\delta_1\)} (M33_2_b.south);
\draw[->] (M33_2.north) -- node[right]{\(\delta_0\)} (M33_2_c.south);

\node[draw, rounded corners=2pt, fill=white, minimum size=0.5cm] (M33_3_a) at (16,12) {};
\node[draw, rounded corners=2pt, fill=white, minimum size=0.5cm] (M33_3_b) at (17,12) {};
\node[draw, rounded corners=2pt, fill=white, minimum size=0.5cm] (M33_3_c) at (18,12) {};
\draw[->] (M33_3.north) -- node[right]{\(\delta_3\)} (M33_3_a.south);
\draw[->] (M33_3.north) -- node[right]{\(\delta_2\)} (M33_3_b.south);
\draw[->] (M33_3.north) -- node[right]{\(\delta_0\)} (M33_3_c.south);

\node[fill=blue, circle, inner sep=1.5pt] (dot_RMid3) at ([xshift = - 10mm]RectMid3.center) {};
\node[fill=blue, circle, inner sep=1.5pt] (dot_Bottom_blue) at ([xshift = - 20mm]Bottom.center) {};

\node[fill=red, circle, inner sep=1.5pt] (dot_Bottom_red1) at ([xshift =  40mm]Bottom.center) {};

\node[fill=red, circle, inner sep=1.5pt] (dot_Bottom_red2) at ([xshift =  20mm]Bottom.center) {};

\node[fill=blue, circle, inner sep=1.5pt] (dot_RMid1_blue) at (RectMid1.center) {};
\node[fill=red, circle, inner sep=1.5pt] (dot_RMid1_red) at ([xshift =  10mm]RectMid1.center) {};

\node[fill=blue, circle, inner sep=1.5pt] (dot_M11) at (M11.center) {};
\node[fill=red, circle, inner sep=1.5pt] (dot_RMid2) at (RectMid2.center) {};

\draw[->, blue, ultra thick] (dot_RMid3) to[bend left=20] (dot_Bottom_blue) ;
\draw[->, blue, ultra thick]
(dot_Bottom_blue) 
to[bend left=20] (dot_RMid1_blue);

\draw[->, blue, ultra thick] (dot_RMid1_blue)
to[bend right=20] (dot_M11);

\draw[->, red, ultra thick] (dot_Bottom_red1) to[bend right=20] (dot_RMid2);
\draw[->, red, ultra thick]
(dot_RMid2)to[bend right=20] (dot_Bottom_red2);
\draw[->, red, ultra thick]
(dot_Bottom_red2)
to[bend right=20] (dot_RMid1_red);




\end{tikzpicture}
    }
    \caption{Another example of distinguishing sequences, following label sequence $(\delta_1, \delta_2, \delta_2)$.
    The paths start at $C_0$ and $C_1$. 
    One reaches a non-leaf cluster that has degree $\delta_0 + \delta_1 + \delta_2 + \delta_3$,
    whereas the other one reaches a leaf of degree $\delta_3$.
    }
    \label{fig:distinguishing-2}
\end{figure}

The following claim states that a distinguishing sequence must include certain labels of low value. This is eventually used in the coupling to show these label sequences are less likely to be followed by the path.

\begin{claim} \label{clm:distinguishing-path}
    Let $d_1, d_2, \ldots, d_m$ be a label sequence that distinguishes between $C_0$ and $C_1$ in $\mT_r$. Then, there must be a subsequence with indices $0 \leq a_1 < a_2 < \ldots a_r \leq m$ such that $d_{a_i} \leq \delta^i$.
\end{claim}
\begin{proof}
    We go over the label sequence and the corresponding paths starting at $C_0$ and $C_1$ chunk by chunk and construct the subsequence along the way. The following invariant is maintained: Let $t$ be the number of indices added so far to the subsequence, let $i$ be the index of the last processed label in the sequence $d$ (initially $i = 0$), and $B^0_i$ and $B^1_i$ be the $i$-th cluster in the corresponding paths starting at $C_0$ and $C_1$. Then, it holds $$\max\left(\tau(B^0_i), \tau(B^1_i)\right) \leq t.$$
    The invariant holds initially as $\tau(B^0_i) = \tau(B^1_i) = t = 0$.
    Now, consider processing the $(i+1)$-th label $d_{i+1}$.
    
    If $d_{i+1} \geq \delta^{t + 2}$, then by the invariant, we can invoke \cref{clm:large-steps}.
    Hence, $B^0_{i+1}$ and $B^1_{i+1}$ are children of $B^0_i$ and $B^1_i$ with identical subtrees.
    Observe that while the paths reside in these subtrees, the sequence cannot distinguish between $C_0$ and $C_1$ because at any point $j$ of the path inside the subtrees $B^0_j$ and $B^1_j$ have the same edge labels adjacent to them.
    Therefore, the paths must exit the subtrees at some point.
    That is, there exists an index $i' > i$ such that $B^0_{i'} = B^0_i$ and $B^1_{i'} = B^1_i$.
    In this case, we process the sequence up to index $i'$ and add nothing to the subsequence.
    The invariant still holds since $t$ does not change and the paths lead to the same clusters $B^0_i$ and $B^1_i$.

    Otherwise, if $d_{i+1} \leq \delta^{t + 1}$, we add $a_{t + 1} = i + 1$ to the sequence. 
    This satisfies the property required of the subsequence:
    $$
    d_{a_{t + 1}} = d_{i+1} \leq \delta^{t + 1}.
    $$
    It remains to show that the invariant holds.
    We prove that as a result of taking the $(i+1)$-th step, the color of the cluster increases by at most $1$.
    That is, 
    $$\tau(B^0_{i+1}) \leq \tau(B^0_{i}) + 1 \qquad \text{and} \qquad \tau(B^1_{i+1}) \leq \tau(B^1_{i}) + 1.$$
    Let us focus on $B^0_{i + 1}$. The claim holds similarly for $B^1_{i+1}$.
    If $B^0_{i+1}$ is the parent of $B^0_{i}$, then $\tau(B^0_{i+1}) \leq \tau(B^0_{i}) - 1$ and the claim holds trivially.
    Otherwise, if $B^0_{i+1}$ is a child of $B^0_{i}$, we can apply \cref{clm:child-color} to conclude
    $$
    \tau(B^0_{i+1}) = \max\left(d(B^0_{i}, B^0_{i + 1}), \tau(B^0_{i}) + 1\right) = \tau(B^0_{i}) + 1.
    $$
    Therefore, the invariant still holds.

    Finally, we prove that $r$ indices are added to the subsequence before this process stops.
    Since the label sequence is distinguishing, the process can stop only when at least one of $B^0_i$ and $B^1_i$ is a leaf.
    That is, when $\max(B^0_i, B^1_i) = r$.
    Therefore, by the invariant, we have the number of indices added to the subsequence is
    \begin{equation*}
    t \geq \max(B^0_i, B^1_i) = r. \qedhere
    \end{equation*}
\end{proof}

\subsection{Full Blueprint} \label{subsec:full-blueprint}

Finally, we present the full blueprint.

\begin{definition}[Full blueprint] 
\label{def:blueprint}
    The full blueprint includes 
    \begin{enumerate}
        \item two instances of the cluster tree $\mT_r$, referred to as $\mT^\zero$ and $\mT^\one$, and
        \item a cluster of dummy vertices $D$.
    \end{enumerate}
    For each cluster $C^\zero \in \mC^\zero$ and the corresponding cluster $C^\one \in \mC^\one$, an edge $(C^\zero, C^\one)$ is added with $d(C^\zero, C^\one) = d(C^\one, C^\zero) = 1$, we refer to these edges as the \emph{main matching} edges.
    For each cluster $C \in \mC^\zero \cup \mC^\one$, an edge $(C, D)$ is added with $d(C, D) = \Delta_r + 1$, and the value of $d(D, C)$ is defined later in \cref{subsec:final-construction}.
    
    We use $\mC_\mB := \mC^\zero \cup \mC^\one \cup \{D\}$ to denote the set of clusters in the full blueprint, and $\mE_\mB$ to denote the edges.
    We define the degree of a cluster $C \in \mC_\mB$ in the blueprint as
    $$
    d_{\mB}(C) = \sum_{C':(C, C') \in \mE_\mB} d(C, C').
    $$
\end{definition}

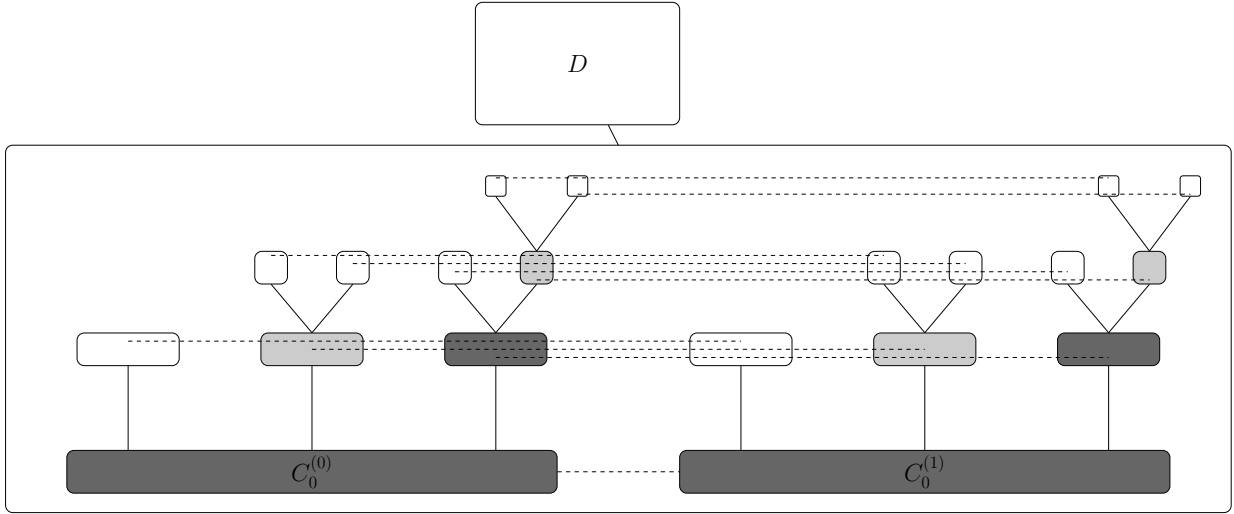
\begin{figure}
    \centering
    \resizebox{\textwidth}{!}{
    \begin{tikzpicture}[>=stealth, transform shape, font=\LARGE]

\begin{scope}[shift={(0,0)}]
  \node[draw, rounded corners=5pt, fill=black!60, 
        minimum width=12cm, minimum height=1cm] 
        (Bottom) at (2.5,0) {$C_0^{(0)}$};
  
  \node[draw, rounded corners=5pt, fill=white, 
        minimum width=2.5cm, minimum height=0.8cm] 
        (RectLeft) at (-2,3) {};
  \node[draw, rounded corners=5pt, fill=black!20, 
        minimum width=2.5cm, minimum height=0.8cm] 
        (RectMid) at (2.5,3) {};
  \node[draw, rounded corners=5pt, fill=black!60, 
        minimum width=2.5cm, minimum height=0.8cm] 
        (RectRight) at (7,3) {};
  
  \draw ([xshift=-4.5cm]Bottom.north) -- (RectLeft.south);
  \draw (Bottom.north) -- (RectMid.south);
  \draw ([xshift=+4.5cm]Bottom.north) -- (RectRight.south);
  
  \node[draw, rounded corners=5pt, fill=white, minimum size=0.8cm] 
        (M_mid1) at (1.5,5) {};
  \node[draw, rounded corners=5pt, fill=white, minimum size=0.8cm] 
        (M_mid2) at (3.5,5) {};
  \draw (RectMid.north) -- (M_mid1.south);
  \draw (RectMid.north) -- (M_mid2.south);
  
  \node[draw, rounded corners=5pt, fill=white, minimum size=0.8cm] 
        (M_right1) at (6,5) {};
  \node[draw, rounded corners=5pt, fill=black!20, minimum size=0.8cm] 
        (M_right2) at (8,5) {};
  \draw (RectRight.north) -- (M_right1.south);
  \draw (RectRight.north) -- (M_right2.south);
  
  \node[draw, rounded corners=2pt, fill=white, minimum size=0.5cm] 
        (M_right2_top1) at (7,7) {};
  \node[draw, rounded corners=2pt, fill=white, minimum size=0.5cm] 
        (M_right2_top2) at (9,7) {};
  \draw (M_right2.north) -- (M_right2_top1.south);
  \draw (M_right2.north) -- (M_right2_top2.south);
\end{scope}

\begin{scope}[shift={(15,0)}]
  \node[draw, rounded corners=5pt, fill=black!60, 
        minimum width=12cm, minimum height=1cm] 
        (Bottom_B) at (2.5,0) {$C_0^{(1)}$};
  
  \node[draw, rounded corners=5pt, fill=white, 
        minimum width=2.5cm, minimum height=0.8cm] 
        (RectLeft_B) at (-2,3) {};
  \node[draw, rounded corners=5pt, fill=black!20, 
        minimum width=2.5cm, minimum height=0.8cm] 
        (RectMid_B) at (2.5,3) {};
  \node[draw, rounded corners=5pt, fill=black!60, 
        minimum width=2.5cm, minimum height=0.8cm] 
        (RectRight_B) at (7,3) {};
  
  \draw ([xshift=-4.5cm]Bottom_B.north) -- (RectLeft_B.south);
  \draw (Bottom_B.north) -- (RectMid_B.south);
  \draw ([xshift=+4.5cm]Bottom_B.north) -- (RectRight_B.south);
  
  \node[draw, rounded corners=5pt, fill=white, minimum size=0.8cm] 
        (M_mid1_B) at (1.5,5) {};
  \node[draw, rounded corners=5pt, fill=white, minimum size=0.8cm] 
        (M_mid2_B) at (3.5,5) {};
  \draw (RectMid_B.north) -- (M_mid1_B.south);
  \draw (RectMid_B.north) -- (M_mid2_B.south);
  
  \node[draw, rounded corners=5pt, fill=white, minimum size=0.8cm] 
        (M_right1_B) at (6,5) {};
  \node[draw, rounded corners=5pt, fill=black!20, minimum size=0.8cm] 
        (M_right2_B) at (8,5) {};
  \draw (RectRight_B.north) -- (M_right1_B.south);
  \draw (RectRight_B.north) -- (M_right2_B.south);
  
  \node[draw, rounded corners=2pt, fill=white, minimum size=0.5cm] 
        (M_right2_top1_B) at (7,7) {};
  \node[draw, rounded corners=2pt, fill=white, minimum size=0.5cm] 
        (M_right2_top2_B) at (9,7) {};
  \draw (M_right2_B.north) -- (M_right2_top1_B.south);
  \draw (M_right2_B.north) -- (M_right2_top2_B.south);
\end{scope}

\draw[dashed] (Bottom) -- (Bottom_B);
\draw[dashed ] ([yshift=2mm]RectLeft.center) -- ([yshift=2mm]RectLeft_B.center);
\draw[dashed ] (RectMid.center) -- (RectMid_B.center);
\draw[dashed ] ([yshift=-2mm]RectRight.center) -- ([yshift=-2mm]RectRight_B.center);
\draw[dashed ] ([yshift=3mm]M_mid1.center) -- ([yshift=3mm]M_mid1_B.center);
\draw[dashed ] ([yshift=1mm]M_mid2.center) -- ([yshift=1mm]M_mid2_B.center);
\draw[dashed ] ([yshift=-1mm]M_right1.center) -- ([yshift=-1mm]M_right1_B.center);
\draw[dashed ] ([yshift=-3mm]M_right2.center) -- ([yshift=-3mm]M_right2_B.center);
\draw[dashed ] ([yshift=2mm]M_right2_top1.center) -- ([yshift=2mm]M_right2_top1_B.center);
\draw[dashed ] ([yshift=-2mm]M_right2_top2.center) -- ([yshift=-2mm]M_right2_top2_B.center);

\draw[ rounded corners=5pt] (-5,-1) rectangle (25,8) node[midway] {};
\node[draw, rounded corners=5pt, fill=white, minimum width=5cm, minimum height=3cm] (D) at (9,10) {$D$};
\draw (10,8) -- (D);

\end{tikzpicture}

    }
    \caption{An illustration of the full blueprint (\Cref{def:blueprint}) for $r = 2$.}
    \label{fig:full-blueprint}
\end{figure}

The degrees and outgoing edge labels are characterized below,
analogous to similar results for the cluster tree.

\begin{claim} \label{clm:non-leaf-degree}
    For any non-leaf cluster $C \in C^\zero \cup C^\one$ of the trees, it holds
    $$
    d_{\mB}(C) = \bar{d}_r + \Delta_r + 2 =: \bar{d}_\mB.
    $$
\end{claim}
\begin{proof}
    Take a cluster $C \in \mC^\zero$ (the case where $C \in \mC^\one$ follows similarly).
    By \cref{clm:non-leaf-degrees-in-tree}, it holds $d_\mT(C) = \bar{d}_r$.
    This accounts for the neighbors of $C$ inside $\mT^\zero$.
    Outside $\mT^\zero$, $C$ has only two neighbors: the dummy cluster $D$, and the corresponding cluster $C'\in \mC^\one$ from the other tree.
    Therefore, it holds:
    \begin{equation*}
        d_\mB(C) = d_\mT(C) + d(C, D) + d(C, C') = 
        \bar{d}_r + (\Delta_r + 1) + 1. \qedhere
    \end{equation*}
\end{proof}

\begin{claim} \label{clm:same-outgoing-labels}
    Take any two clusters $C, C' \in \mC_\mB$ in the full blueprint.
    If $d_\mB(C) = d_\mB(C')$, then $C$ and $C'$ have the same set of outgoing edge-labels.
\end{claim}
\begin{proof}
    The dummy cluster $D$ has a degree strictly larger than all the other clusters.
    For all other clusters $C, C' \in \mC_\mB \setminus \{D\}$,
    it holds
    $$
    d_\mT(C) = d_\mB(C) - (\Delta_r + 2) = d_\mB(C') - (\Delta_r + 2)
    = d_\mT(C').
    $$
    Therefore, by \cref{clm:same-outgoing-labels-in-tree},
    $C$ and $C'$ have the same outgoing labels inside their tree.
    Outside, they both have an outgoing label $1$ to the corresponding cluster of the other tree,
    and an outgoing label $\Delta_r + 1$ to the dummy cluster $D$.
    Hence, they have the same set of outgoing labels overall.
\end{proof}

\subsection{Final Construction} \label{subsec:final-construction}

In this section, we show how to construct the final input distribution using the blueprint.

\begin{definition}[Input Distribution] \label{def:final-graph}
    For every cluster $C \in \mC_\mB$, there is a size $\card{C}$ associated with it.
    The vertices of the final graph $G$ consist of $\card{C}$ vertices for every cluster $C \in \mC_\mB$, and all the vertex labels are chosen at random.
    Then, for each edge $(X, Y) \in \mE_\mB$, a bipartite subgraph $H$ is added between the vertices of $X$ and $Y$,
    where $H$ is chosen at random from all bipartite subgraphs between $X$ and $Y$ that are $d(X, Y)$-regular on the $X$ side, and $d(Y, X)$-regular on the $Y$ side.
    The cluster sizes are selected in a way that $\card{X}d(X, Y) = \card{Y}d(Y, X)$
    to ensure that at least one such subgraph $H$ exists.
    We say that each edge in $H$ has label $d(X, Y)$ in the $X$-to-$Y$ direction, and label $d(Y, X)$ in the $Y$-to-$X$ direction.
    Finally, the adjacency list of every vertex is permutated uniformly at random.
\end{definition}

Now, we specify the cluster sizes.
The corresponding clusters of the two trees $\mT^\zero$ and $\mT^\one$ shall have the same size.
We use $N_0$ to denote the size of the root clusters $C_0^\zero$ and $C_0^\one$.
Then, for every non-root cluster $C \in \mC^\zero \cup \mC^\one$, we let $\card{C} = \card{p_C} / \delta$.
Let $N := \sum_{C\in\mC^\one}\card{C}$ be the total number of vertices in one of the trees. Finally, we let $\card{D} = 2\epsilon N$, where $\epsilon$ is a small constant determined later.

We also specify the edge labels $d(D, C)$ that were deferred here from \cref{subsec:cluster-trees}.
For a cluster $C \neq D$ we let 
$$
d(D, C) = \frac{\card{C}}{\card{D}} d(C, D) = \frac{\Delta_r + 1}{2\epsilon N}  \card{C}.
$$

\begin{claim} \label{clm:dummy-degree} \label{clm:blueprint-max-degree}
    In the full blueprint, it holds
    $ d_\mB(D) = (\Delta_r + 1)/\epsilon.$
    This is also the maximum degree $\Delta$ in the final graph $G$.
\end{claim}
\begin{proof}
    We can break up the degree of $D$ into two parts:
    \begin{equation}
    d_\mB(D) = \sum_{C:(D, C)\in\mE_\mB} d(D, C)
    = \sum_{C \in \mC^\zero} d(D, C) + \sum_{C \in \mC^\one} d(D, C).
    \label{eq:clm-dummy-degree-1}
    \end{equation}
    Let us focus on the first part:
    \begin{align*}
        \sum_{C \in \mC^\zero} d(D, C)
        &=  \sum_{C \in \mC^\zero} \frac{\Delta_r + 1}{2\epsilon N}  \card{C} \\
        &= \frac{\Delta_r + 1}{2\epsilon N} \sum_{C \in \mC^\zero} \card{C} \\
        &= \frac{\Delta_r + 1}{2\epsilon}
    \end{align*}
    Similarly, it can be seen $\sum_{C \in \mC^\one} d(D, C) = (\Delta_r + 1)/2\epsilon$. Plugging back into \eqref{eq:clm-dummy-degree-1}, we get:
    \begin{equation*}
    d_\mB(D) = \frac{\Delta_r + 1}{2\epsilon} + \frac{\Delta_r + 1}{2\epsilon} = \frac{\Delta_r + 1}{\epsilon}. \qedhere
    \end{equation*}
\end{proof}

\begin{claim} \label{clm:value-of-N}
    It holds
    $$
    N \leq N_0 \left(1 + \frac{r + 1}{\delta - (r + 1)}\right)
    $$
\end{claim}
\begin{proof}
    The depth of a cluster in the tree is its distance from the root $C_0$.
    Observe that any cluster $C$ at depth $i$, has at most $r + 1$ child clusters, each of size at most $\card{C}/\delta$.
    Therefore, the total cluster size at depth $i + 1$ is at most $\frac{r+1}{\delta}$ times that of depth $i$, and we have:
    \begin{equation*}
        N \leq \sum_{i=0}^r N_0 \left(\frac{r+1}{\delta}\right)^i
        \leq N_0 \left(1 + \frac{r + 1}{\delta - (r + 1)}\right). \qedhere
    \end{equation*}
\end{proof}

We set the values of the parameters here.
As a result, the total number of vertices shall be $n = (2 + 2\epsilon)N = \Theta(N) = \Theta(N_0)$ (see also \cref{clm:value-of-N}).
Let $c$ be the approximation ratio
for which we prove the lower bound, and $\epsilon := c/6$.
Let $r = \log \Delta/\log\log\Delta$, and $\delta$ be determined with $\Delta = \frac{\delta^{r+1} + 1}{\epsilon}$.
Note that as a result, it holds $\delta = \omega(\log\Delta/\log\log\Delta)$, in particular, it holds $\delta \geq \left(\frac{3}{c}+1\right)(r+1)$ for large enough $\Delta$.
Let $\kappa = \Delta^{(1/6)\log\Delta/\log\log\Delta}$,
so that $\Delta^{r/6} = \kappa$.
Finally, let $N_0$ be large enough such that $\kappa^2\Delta^2 = o(N_0)$,
i.e.\ $n$ must be large enough such that $\Delta = 2^{O(\sqrt{\log n \log \log n})}$.
\section{Maximum Matching Lower Bound}
\label{sec:matching-lb}

To prove the lower bound, we show that on the input distribution of \cref{def:final-graph}, any deterministic non-adaptive LCA requires $\Delta^{\Omega(\log\Delta/\log\log\Delta)}$ queries to compute a constant approximation of the maximum matching.
Then, the lower bound for randomized algorithms follows from a direct application of Yao's min-max principle.

To prove the lower bound in the deterministic case,
let us call the $C_0^\zero$--\,$C_0^\one$ edges \emph{significant} since they make up a constant fraction of the maximum matching,
and call the $C_0^\zero$--\,$C_1^\zero$ edges \emph{misleading} since
the size of the maximum matching among these edges is $o(1)\cdot \mu(G)$, but, they look roughly similar to significant edges.
We argue that the algorithm cannot distinguish between significant and misleading edges.
Therefore, as long as it is outputting a valid matching, it cannot output many of the significant edges, and thus has a poor approximation ratio.
This statement is formalized below.

Fix the edge set of the graph. Then, considering the randomness of the vertex IDs and the permutations of the adjacency lists, let $x_e$ be the probability that an edge $e$ is in the output of the LCA.

\begin{claim} \label{clm:small-fractional-matching}
    Take a deterministic non-adaptive LCA with query complexity $\kappa = \Delta^{\Omega(\log\Delta/\log\log\Delta)}$.
    For any edge $e \in E(G)$, let $x_e$ be the probability that the LCA includes $e$ in the output.
    Then, for any significant edge $e$, i.e.\ an edge between the vertices of $C_0^\zero$ and $C_0^\one$, it holds
    $x_e = o(1)$.
\end{claim}

The proof is deferred to later in the section.
Assuming the claim, we prove our main results.

\begin{lemma} \label{lem:deterministic-lb}
    Any \emph{deterministic} non-adaptive LCA that given a graph $G$ from the input distribution (as defined in \cref{def:final-graph}), $O(1)$-approximates the maximum matching with constant probability, requires $\kappa = \Delta^{\Omega(\log \Delta / \log \log \Delta)}$ queries.
\end{lemma}

\begin{proof}
    Taking an arbitrary constant $c \in (0, 1)$, we prove the lemma for any algorithm that outputs a matching of size at least $c\cdot\mu(G)$ in expectation.
    This also implies the lower bound for any algorithm that computes a $2c$ approximation with probability $\frac{1}{2}$.

    Since $x_e$ is the probability of each edge appearing in the output,
    by linearity of expectation,
    the size of the output matching is $\sum_e x_e$ in expectation.
    We upper-bound this sum.
    Observe that $V(G) \setminus (C_0^\zero \cup C_0^\one)$ is a vertex cover for the insignificant edges.
    Therefore, the sum of $x$ over these edges is at most
    $$
    \card{V(G) \setminus (C_0^\zero \cup C_0^\one)} = 2(N - N_0) + 2\epsilon N.
    $$
    Furthermore, the sum of $x$ over the significant edges is at most $o(N_0)$ since there are $N_0$ significant edges, and $x_e = o(1)$ for each significant edge $e$ (\cref{clm:small-fractional-matching}).
    Hence, the expected size of the output matching is
    $$
    \sum_e x_e = 2(N - N_0) + 2 \epsilon N +  o(N_0).
    $$
    On the other hand, the main-matching edges match all the non-dummy vertices. Therefore,
    $$
    \mu(G) \geq N.
    $$
    
    Putting the two together, the approximation ratio is at most
    \begin{align*}
    \frac{2(N - N_0) + 2 \epsilon N  + o(N_0)}{N}
    &\leq \frac{2(N-N_0)}{N_0} + 2\epsilon + o(1) \\
    &= \frac{r+1}{\delta - (r+1)} + 2\epsilon + o(1),
    \end{align*}
    where the second inequality follows from 
    $N \leq N_0 \left(1 + \frac{r + 1}{\delta - (r + 1)}\right)$ (\cref{clm:value-of-N}).
    Recall the values of the parameters $\delta \geq \left(\frac{3}{c}+1\right)(r+1)$ and $\epsilon = c/6$.
    Plugging these values, the approximation ratio is at most
    $$
    \frac{c}{3} + \frac{c}{3} + o(1) < c,
    $$
    which concludes the proof.
\end{proof}

We are now ready to prove our main theorem.

\begin{proof}[Proof of \cref{thm:main}]
    Follows from a direct application of Yao's min-max principle to \cref{lem:deterministic-lb}.
\end{proof}

The remainder of this section is devoted to the proof of \cref{clm:small-fractional-matching}.
First, we show that the subgraph induced by the vertices that the algorithm discovers is a tree with high probability (\cref{clm:edge-prob,clm:no-cycle}).
Then, we show that starting from a significant edge or a misleading edge, the tree that the algorithm discovers is similar in distribution (\cref{clm:path-coupling,clm:tree-coupling}).
As a result, we can show that the algorithm outputs each significant edge with probability at most $o(1)$. Hence, the approximation ratio suffers.

Our approach to prove the algorithm does not explore any cycles is the same as \cite{BehnezhadRR-FOCS23} (Lemmas 5.1 and 5.2). 
We include the argument here for completeness.

Intuitively, the following claim shows that since the graph is sparse ($\Delta = o(N_0)$), if the algorithm has not discovered many edges, the remainder of each of the regular bipartite subgraphs behaves like an \ER{} graph. Thus, we can bound the probability of existence for each edge.

\begin{claim} \label{clm:edge-prob}
    Consider the process of discovering the graph according to a fixed sequence of query instructions, and assume the number of queries $Q \leq \kappa = \Delta^{\Omega(\log\Delta/\log\log\Delta)}$.
    At any point, take a pair of vertices $u$ and $v$ such that no edge between them has been discovered.
    Conditioned on the edges discovered so far,
    the probability that there exists an edge between $u$ and $v$ is at most $O\left(\frac{\Delta^2}{N_0}\right)$.
\end{claim}
\begin{proof}
    Let $X \in \mC_\mB$ be the cluster of $u$ and $Y \in \mC_\mB$ be the cluster of $v$.
    For simplicity, let $d_X$ denote $d(X, Y)$, and $d_Y$ denote $d(Y, X)$.
    Observe that both $\card{X}$ and $\card{Y}$ are at least $N_0 / \delta^r \geq N_0/\Delta$.
    If there is no edge between $X$ and $Y$ in $\mE_\mB$, then there is no edge between $u$ and $v$ in $G$, and the claim holds.
    Otherwise, let $F$ be the set of discovered edges between $X$ and $Y$.
    We can disregard other discovered edges in our conditioning as the set of edges between $X$ and $Y$ is chosen independently from the remainder of the graph.
    That is, we study the probability of $(u, v) \in E(G)$ conditioned on the fact that $F \subseteq E(G)$.
    Roughly, this probability is at most $\frac{d_X}{\card{Y}} = \frac{d_Y}{\card{X}}$, which is the density parameter of a \ER{} graph that has expected degree $d_X$ on the $X$ side and $d_Y$ on the $Y$ side.

    Let $\mG$ be the set of bipartite graphs between $X$ and $Y$
    that are $d_X$-regular on the $X$ side and $d_Y$-regular on the $Y$ side, \emph{and include $F$ in their edge set}.
    Let $\mG_0 \subseteq \mG$ be the subset of graphs that include edge $(u, v)$, and $\mG_1$ be the subset of graphs that do not.
    We bound the probability of interest as follows:
    $$
    \Pr[(u,v) \in E(G) \mid F \subseteq E(G)]
    = \frac{\card{\mG_0}}{\card{\mG}}
    \leq \frac{\card{\mG_0}}{\card{\mG_1}}.
    $$

    To bound $\card{\mG_0}/\card{\mG_1}$, we define a relation between graphs of $\mG_0$ and $\mG_1$, such that each graph in $\mG_0$ is related to at least $\Omega(d_Y \card{Y}) \geq \Omega(d_Y \cdot N_0/\Delta)$
    graphs of $\mG_1$,
    and each graph in $\mG_1$ is related to at most $O(d_Y \cdot \Delta)$
    graphs of $\mG_0$.
    Then, it follows
    $$
    \card{\mG_0} \cdot \Omega(d_Y \cdot N_0/\Delta)
    \leq \card{\mG_1} \cdot O(d_Y \cdot \Delta),
    $$
    and therefore,
    $$
    \frac{\card{\mG_0}}{\card{\mG_1}} = O\left(\frac{\Delta^2}{N_0}\right).
    $$

    First, we define the \emph{two-switch} operation.
    Take four vertices $a, b, c, d$ in a graph $G$ such that the subgraph induced by these four vertices includes only two edges $(a, b)$ and $(c, d)$.
    Then, performing a two-switch on $(a, b)$ and $(c, d)$ results in a new graph $G'$ such that
    $$
    E(G') = E(G) \setminus \{(a, b), (c,d)\} \cup \{(a, d), (b, c)\}.
    $$

    The relation is defined as follows:
    For a graph $G \in \mG_0$, consider all edges $(x, y) \notin F$ between $X$ and $Y$ such that a two-switch on $(u, v)$ and $(x, y)$ is defined.
    That is, there is no edge $(u, y)$ and no edge $(x, v)$ in $E(G)$.
    Let $G' = G \oplus (x, y)$ denote the graph obtained by applying a two-switch on $(u, v)$ and $(x, y)$,
    and let $G$ be related to all such graphs $G'$.
    Since $(u, v), (x, y) \notin F$, it still holds $F \subseteq E(G')$.
    Also, the two-switch operation deletes edge $(u, v)$ and does not change the degrees.
    Therefore, $G'$ is a member of $\mG_1$.
    It remains to bound the number of relations on each side.

    Take a graph $G \in \mG_0$.
    We count all the candidates $(x, y)$ for a two-switch operation.
    Since $G[X; Y]$ is $d_X$-regular on the $X$ side, there are $\card{Y} - d_X$ vertices $y \in Y$ that are not adjacent to $u$.
    Hence, as the graph is $d_Y$-regular on the $Y$ side, there are $(\card{Y} - d_X)d_Y$ edges $(x, y)$ such that $y$ is not adjacent to $u$.
    Among these edges, at most $\kappa$ of them are discovered (i.e.\ in $F$), and at most $d_Xd_Y$ of them have $x$ adjacent to $v$.
    Therefore, the number of two-switch candidates, and as a result, the number of graphs $G' \in \mG_1$ that $G$ is related to, is at most
    $$
    (\card{Y} - d_X)d_Y - \kappa - d_Xd_Y
    \geq \Omega(d_Y\cdot N_0/\Delta),
    $$
    where the inequality follows from $\card{Y} \geq N_0/\Delta$,
    $\kappa = o(d_Y \cdot N_0/\Delta)$,
    and $d_X \leq \Delta = o(N_0/\Delta)$.

    For the last part of the proof, take a graph $G' \in \mG_1$.
    To bound the number of graphs $G \in \mG_0$ related to $G'$,
    we characterize the pairs $(x, y)$ such that there exists a graph $G$ with $G' = G \oplus (x, y)$.
    For the equality to hold, we must have 
    $$
    (u, y), (x, v) \in E(G') \qquad \text{and} \qquad (u, v), (x, y) \notin E(G').
    $$
    By the first part above, we can upper-bound the number of valid pairs $(x, y)$ by $d_Xd_Y$ (there are at most $d_X$ choices for $y$ and $d_Y$ choices for $x$).
    Furthermore, for each $(x, y)$ there exists exactly one $G$ such that $G' = G \oplus (x, y)$ since $G$ can be explicitly expressed as $$E(G) = E(G') \setminus\{(u, y), (x, v)\} \cup \{(u, v), (x, y)\}.$$
    Therefore, there are at most $d_Xd_Y$ graphs $G \in \mG_0$ related to $G'$. Combined with the previous paragraph, this implies $\frac{\card{\mG_0}}{\card{\mG_1}} = O\left(\frac{\Delta^2}{N_0}\right),$ which gives the same bound on the probability, and thus concludes the proof.
\end{proof}

Since the graph is sparse ($\Delta = o(N_0)$),
we can use the claim above to show that if the algorithm does not make many queries, it will find no cycles.

\begin{claim} \label{clm:no-cycle}
    Take a deterministic non-adaptive LCA with at most $\kappa = \Delta^{\Omega(\log\Delta/\log\log\Delta)}$ queries,
    and let $H$ be the subgraph induced by the discovered vertices when prompted with an edge $(u, v)$.
    Then, $H$ is a tree with high probability.
\end{claim}
\begin{proof}
    Consider the process of discovering the graph according to the query instructions.
    Let $H_i$ be the subgraph discovered after $i$ queries, where $E(H_0) = \{(u, v)\}$.
    Let $E_i$ be the event that $i$ is the first step for which the subgraph induced by the discovered vertices, i.e.\ the non-singleton vertices of $H_i$, is not a tree.
    We upper-bound $\Pr(E_i)$ by $i \cdot O( \Delta^2/N_0)$ and apply the union bound to prove the claim.

    Consider step $i$.
    If the induced subgraph in the previous step was not a tree, then $E_i$ does not hold.
    Hence, we can restrict our attention to the case where the previous induced subgraph was a tree.
    As a result of the $i$-th query $a_i$, a new neighbor $v$ of $u_{a_i}$ might be revealed (that is, if the degree of $u_{a_i}$ is larger than $b_i$, the queried index in the adjacency list).
    Vertex $v$ must have been previously undiscovered since the previous induced subgraph was a tree.
    Therefore, event $E_i$ is equivalent to $v$ having at least one edge to another previously discovered vertex $u_j$, where $j \neq a_i$.
    There are $i$ discovered vertices other than $u_{a_i}$.
    By \cref{clm:edge-prob}, for each of these vertices, the probability of being adjacent to $v$ is at most $O(\Delta^2 / N_0)$.
    Taking the union bound over the discovered vertices, we get
    $$
    \Pr[E_i] \leq i \cdot O(\Delta^2/N_0).
    $$
    Finally, taking the union bound over all $i \leq \kappa$,
    implies that the subgraph induced by all the discovered vertices contains a cycle with probability at most
    \begin{equation*}
    O\left(\frac{\kappa^2 \Delta^2}{N_0}\right) = o(1). \qedhere
    \end{equation*}
\end{proof}

Given \cref{clm:no-cycle}, for the remainder of the argument, we condition on the high-probability event that the subgraph induced by the discovered vertices is a tree.
Recall that the algorithm computes the output, based on the sequence of vertices $u_1, \ldots, u_Q$ returned by the queries,
and their degrees $d_1, \ldots, d_Q$ which is also revealed to the algorithm.
As the induced subgraph is a tree,
we can assume without loss of generality that the algorithm computes the output based on the sequence of degrees $d_1, \ldots, d_Q$ (where degree $\bot$ is used for vertices $u_i = \bot$ that do not exist).
Hereafter, we refer to the query instructions as the query tree,
and when a vertex is returned as a result of an adjacency list query to a vertex $u$, we say the child of $u$ in the tree has been revealed.

To show that the algorithm cannot distinguish between a significant edge (i.e.\ a $C_0^\zero$--\,$C_1^\zero$ edge) and a misleading edge (i.e.\ a $C_0^\zero$--\,$C_0^\one$ edge), we show that the distribution of the degree sequence returned by the queries (i.e.\ what the algorithm observes) is almost equal between the two cases.
As a first step, we prove a similar statement for a path in the query tree. 

Intuitively, the claim is proven by showing that the path can reveal information only until it reaches $D$, which happens with constant probability in every step.
Then, the probability that the path reveals useful information before reaching $D$ is bounded by utilizing distinguishing label sequences as characterized in \cref{clm:distinguishing-path}.
To make the proof simpler, we define one of its main components here.

\begin{definition}
\label{def:label-coupling}
    Given two vertices $u$ and $u'$ with $\deg(u) = \deg(u')$, 
    consider the two processes of choosing a random neighbor $v$ of $u$,
    and choosing a random neighbor $v'$ of $u'$,
    the \emph{label-based coupling of the neighbors} of $u$ and $u'$ is defined as follows:

    Let $C, C' \in \mC_\mB$ be the clusters of $u$ and $u'$.
    It holds $d_\mB(C) = \deg(u) = \deg(u') = d_\mB(C')$.
    Therefore, by \cref{clm:same-outgoing-labels},
    clusters $C$ and $C'$ have the same set of outgoing edge-labels.
    That is, for every neighbor cluster $B$ of $C$, there is a neighbor $B'$ of $C'$ with $d(C, B) = d(C', B')$.
    Hence, $u$ has $d(C, B)$ neighbors $v \in B$,
    and $u'$ has the same number of neighbors $v' \in B'$.
    Each vertex $v \in B$ is coupled with a vertex $v' \in B'$.
    We also say the edges $(u, v)$ and $(u', v')$ are coupled.
\end{definition}

\begin{observation}
    In the label-based coupling (as defined in \cref{def:label-coupling}),
    it holds (1) the coupled edges $(u, v)$ and $(u', v')$ have the same labels,
    (2) traversing an edge of label $x$ has probability $\frac{x}{\deg(u)} = \frac{x}{\deg(u')}$,
    and (3) when $u$ and $u'$ are in non-dummy clusters, their edges to the dummies are coupled to each other, and so are their main-matching edges.
\end{observation}

\begin{claim} \label{clm:path-coupling}
    In the query tree, fix a path $P$ starting at $u$, where $(u, v)$ is the prompted edge.
    Let $\bd_0$ be the random variable denoting the sequence of degrees discovered by the algorithm when the prompted edge is between $C_0^\zero$ and $C_0^\one$ (significant),
    and $\bd_1$ denote the sequence of degrees discovered when the edge is between $C_0^\zero$ and $C_1^\zero$ (misleading).
    It holds:
    $$
    \dtv(\bd_0, \bd_1) \leq O\left(\frac{1}{\kappa^2}\right).
    $$
\end{claim}
\begin{proof}
    Let the random variables $v_0 = u, v_1, \ldots, v_L$ denote the vertices discovered by the algorithm on the query path $P$, when $(u, v)$ is between $C_0^\zero$ and $C_0^\one$.
    Define $v'_0, v'_1, \ldots v'_L$ similarly for the case where $(u, v)$ is between $C_0^\zero$ and $C_1^\zero$.
    To prove the claim, we couple $\{v_i\}_i$ and $\{v'_i\}_i$ such that 
    $\deg(v_i) = \deg(v'_i)$ for all $0 \leq i \leq L$ with a probability larger than $1 - \frac{2}{\kappa^2}$.

    The coupling is defined iteratively, by revealing the vertices of the two paths step by step.
    The paths start at the prompted edges.
    Throughout the coupling, we try to maintain the invariant that the two sequences of degrees revealed on the path so far are equal.
    If this attempt fails, we say that the coupling has failed and the rest of the process is coupled arbitrarily.
    To prove the claim, we show that the coupling fails overall with probability at most $2/\kappa^2$.
    On a high level, the two paths can be thought of as walks between the clusters of $\mC_\mB$ which follow the same edge-labels.

    Assume that so far the vertices $v_0, v_1, \ldots v_k$
    and $v'_0, v'_1, \ldots v'_k$ have been revealed, and the invariant holds.
    Specifically, $\deg(v_k) = \deg(v'_k)$.
    Let $b$ be the index of the adjacency list that the algorithm queries.
    If $b > \deg(v_k)$, then in both cases, the $b$-th element of the adjacency list is $\bot$.
    As a result, all the subsequent vertices in the paths are also $\bot$, and the paths are coupled without fail.
    Otherwise, since we permutate the adjacency list at random, the $b$-th element in the adjacency list of $v_k$ or $v'_k$ is a neighbor chosen uniformly at random.
    Hence, we can use the label-based coupling of the neighbors of $v_k$ and $v'_k $(\cref{def:label-coupling}) to determine $v_{k+1}$ and $v'_{k+1}$.
    For now, ignore the fact that one of the neighbors of $v_k$ is $v_{k-1}$. This subtlety is addressed at the end of the proof.

    Now we bound the probability of the coupling failing, i.e.\
    the probability that for some $k$, when $v_k$ and $v'_k$ are revealed according to the coupling, it holds $\deg(v_k) \neq \deg(v'_k)$.
    Assuming the coupling has not failed up to step $k$,
    if it happens that $v_{k+1} \in D$, then it must also hold $v'_{k+1} \in D$.
    In this case, we say the coupling has converged, and it will not fail in the future.
    This is true since once the paths reach the same cluster (in this case $D$), following the same outgoing edge-labels leads them to the same clusters in any subsequent step $k' > k$,
    and thus $\deg(v_{k'}) = \deg(v'_{k'})$.

    We bound the probability of not reaching $D$ in the first $\ell = 2 \log \kappa$ steps.
    Recall that for any cluster $C \neq D$, it holds $d(C, D) = \Delta_r + 1$. Therefore, when $v_k \in C$, the probability of moving to $D$ is:
    $$
    \frac{d(C, D)}{d_\mB(C)} = \frac{\Delta_r + 1}{d_\mT(C) + \Delta_r + 2} \geq \frac{1}{2},
    $$
    where the last inequality follows from $d_\mT(C) \leq \Delta_r$.
    Therefore, the probability of not moving to the dummy cluster in the first $\ell$ steps is at most $2^{-\ell} \leq 1/\kappa^2$.

    Separately, we bound the probability of the coupling failing in the first $\ell$ steps.
    Consider the prefix of the path that does not contain $D$.
    As argued before, the coupling can only fail in this part.
    In addition, ignore the steps where a main-matching edge is traversed, i.e.\ when the path simply moves to the corresponding cluster in the other tree.
    Considering the remaining subsequence $P'$, seeing as the coupling fails, the labels of the traversed edges (which are the same for the two paths) form a distinguishing sequence (\cref{def:distinguishing-sequence}).
    Therefore, by \cref{clm:distinguishing-path}, there must be a subsequence of labels $l_1, \ldots, l_r$
    such that $l_i \leq \delta^i$.
    We call the subsequence and its steps \emph{critical}.
    We bound the probability that a critical subsequence exists in the first $\ell$ steps.
    
    Observe that there are $\binom{\ell}{r}$ possible positions for the elements of the subsequence.
    Also, the critical steps must be taken from non-leaf clusters of the trees.
    Therefore, the probability of taking a critical step from a cluster $C$, and traversing an edge of label $l_i$ is
    $$
    \frac{l_i}{d_\mB(C)} = \frac{l_i}{\bar{d}_r + \Delta_r + 1} \leq \frac{l_i}{\delta^{r+1}}.
    $$
    Hence, the probability of the path failing in the first $\ell$ steps is at most
    \begin{align*}
        \binom{\ell}{r} \prod_{1\leq i\leq r}\frac{l_i}{\delta^{r+1}}
        &\leq \ell^r \prod_{1\leq i\leq r}\frac{\delta^i}{\delta^{r+1}} \\
        &= \ell^r \frac{\delta^{r(r+1)/2}}{\delta^{r(r+1)}} \\
        &= \left(\frac{\ell}{\delta^{(r+1)/2}}\right)^r \\
        &\leq \frac{1}{\kappa^2}.
    \end{align*}
    To see why the last step holds, recall that 
    $$\ell = 2 \log \kappa, \qquad 
    \kappa = \Delta^{\Omega(\log\Delta/\log\log\Delta)}, \qquad \text{and} \qquad
    \Delta = \frac{\delta^{r+1} + 1}{\epsilon}.
    $$
    Therefore,
    \begin{align*}
        \left(\frac{\ell}{\delta^{(r+1)/2}}\right)^r
        &\leq \left(\frac{\ell}{(\Delta/2\epsilon)^{1/2}}\right)^r \\
        &\leq \left(\frac{\Omega(\log^2\Delta/\log\log\Delta)}{(\Delta/2\epsilon)^{1/2}}\right)^r \\
        &\ll \left(\frac{1}{\Delta^{1/3}}\right)^r \\
        &= \frac{1}{\kappa^2}.
    \end{align*}

    Putting things together, we have shown that the probability of the coupling failing in the first $\ell$ steps is at most $1/\kappa^2$.
    Also, the probability of the paths not reaching $D$ in the first $\ell$ steps,
    and hence the probability of the coupling failing after the first $\ell$ steps,
    is at most $1/\kappa^2$.
    Therefore, the overall probability of the coupling failing is at most $2/\kappa^2$.

    It remains to address the issue of going back on the path.
    That is, while querying neighbors of $v_k$, the algorithm might rediscover the previous vertex of the path, i.e.\ $v_{k+1} = v_{k-1}$.
    To do so, we slightly modify the coupling.
    We couple this case with $v'_{k+1} = v'_{k-1}$, since both instances are observing repeated vertex labels.
    This causes an issue if the label of $(v_{k}, v_{k+1})$ is not the same as $(v'_{k}, v'_{k+1})$,
    as it causes the remainder of the outgoing labels of $v_k$ and $v'_k$
    to differ by one edge, making it infeasible to couple based on edge label.
    We show that through a careful treatment of the coupling, this discrepancy is always limited to at most one edge, and that it does not pose any challenges to the rest of the argument.

    Starting at $v_1 \in C_0^\zero$, $v_2 \in C_0^\one$ in one case,
    and at $v'_1 \in C_0^\zero$, $v'_2 \in C_1^\zero$ in the other,
    we have coupled the cases $v_3 = v_1$ and $v'_3 = v'_1$.
    The remainder of outgoing labels of $v_2$ and $v'_2$ are the same,
    except $v_2$ has $\delta$ remaining edges with label $\delta$ whereas $v'_2$ has $\delta - 1$,
    and $v'_2$ has one main-matching edge and $v_2$ does not.
    We call the main-matching edge of $v'_2$ and one of the edges of $v_2$ that has label $\delta$ \emph{special}, couple the special edge together, and couple the remaining edges as before based on equal labels.
    
    More generally, if $(v_k, v_{k-1})$ has label $s$ and $(v'_k, v'_{k-1})$ has label $s'$,
    then among the remaining edges, 
    $v_k$ has one fewer edge of label $s$ compared to $v'_k$,
    and $v'_k$ has one fewer edge of label $s'$ compared to $v_k$.
    Therefore, we choose one outgoing edge of $v_k$ with label $s'$,
    and one outgoing edge of $v'_k$ with label $s$ as special.
    Then, we couple the special edges together, and couple the remaining edges based on equal labels.
    We modify the definition of critical subsequence and allow it to contain special edges. 
    This does not interfere with the argument since (1) the special edges are traversed with probability at most $\frac{1}{\Delta}$, (2) the vertices do not have special edges while they are inside identical subtrees of the blueprint, and (3) traversing special edges increases the maximum color of the clusters by at most $1$ (analogous to the proof of \cref{clm:distinguishing-path} for details).
    As a result, even with the special edges, the coupling fails with probability at most $\frac{2}{\kappa^2}$ which concludes the proof.
\end{proof}

Having shown each path in the query tree does not reveal much information about the prompted edge,
we prove the same thing about the query tree as a whole.

\begin{claim} \label{clm:tree-coupling}
    Let $\bd_0$ be the random variable denoting the sequence of degrees discovered by the query tree when the prompted edge is between $C_0^\zero$ and $C_0^\one$,
    and $\bd_1$ denote the sequence of degrees discovered when the edge is between $C_0^\zero$ and $C_1^\zero$.
    Then,
    $$
    \dtv(\bd_0, \bd_1) \leq O\left(\frac{1}{\kappa}\right) = o(1).
    $$
\end{claim}

\begin{proof}
    Since the coupling of \cref{clm:path-coupling} is defined iteratively,
    it can be implemented for all the nodes of the query tree simultaneously.
    Then, the coupling fails on each path starting at the root (the prompted edge) with probability at most $O(1/\kappa^2)$.
    Taking the union bound over all the possible $\kappa$ paths,
    the coupling fails with probability at most $O(1/\kappa)$.
    That is, the sequence of degrees that the algorithm observes is different in the two cases with probability at most $O(1/\kappa)$.
    This concludes the proof.
\end{proof}

Finally, we are ready to prove that the algorithm is not very likely to include each significant edge in the output.

\begin{proof}[Proof of \cref{clm:small-fractional-matching}]
    The LCA always outputs a matching, i.e.\ for any vertex $u$,
    at most one adjacent edge is in the output.
    Therefore, the probabilities $x_e$ of each edge $e$ being in the output form a fractional matching.
    Now, fix any vertex $u \in C_1^\zero$ and consider its $\delta$ edges to $C_0^\zero$.
    Since we have assumed that the subgraph induced by the discovered vertices is a tree, the tree queried from each of these edges has the exact same distribution.
    Therefore, all these edges have the same probability of being added to the output $x_e \leq \frac{1}{\delta}$.

    By \cref{clm:tree-coupling},
    the distribution of the tree queried from a $C_0^\zero$--\,$C_0^\one$ edge
    is at most $o(1)$-far from that of a $C_0^\zero$--\,$C_1^\zero$ edge.
    Therefore, for any edge $e$ between $C_0^\zero$ and $C_0^\one$,
    it holds
    \begin{equation*}
    x_e \leq \frac{1}{\delta} + o(1) = o(1). \qedhere
    \end{equation*}
\end{proof}

\section{Vertex Cover Lower Bound} \label{sec:vc}

Similar to \cref{sec:matching-lb}, we prove a $\Delta^{\Omega(\log\Delta/\log\log\Delta)}$ lower bound for any deterministic algorithm that computes a constant approximation of the minimum vertex cover against our hard input distribution.
Then, we apply Yao's min-max principle to prove the lower bound for randomized LCAs.
We use a slightly modified version of \cref{def:final-graph} as the input distribution.
Namely, we remove all the main-matching edges.
That is, the edge between corresponding clusters of the two trees.

Let $S_0 := C_0^\zero \cup C_0^\one$,
and $S_1 := C_1^\zero \cup C_1^\one$.
The minimum vertex cover includes all vertices except those of $S_0$.
Yet, we show that the algorithm cannot distinguish between $S_0$ and $S_1$.
Therefore, to cover the edges between $S_0$ and $S_1$, the algorithm must include about half of $S_0$ in the output, which is suboptimal.

\begin{claim} \label{clm:large-fractional-cover}
    Take a deterministic non-adaptive LCA with query complexity $\kappa = \Delta^{\Omega(\log\Delta/\log\log\Delta)}$ for the minimum vertex cover problem.
    For any vertex $u \in V(G)$, let $y_u$ be the probability that the LCA includes $u$ in the output.
    Then, for any vertex $u \in S_0$, it holds
    $y_u \geq \frac{1}{2} - o(1)$.
\end{claim}
\begin{proof}
    Since the LCA always outputs a valid vertex cover,
    for every edge $(u, v)$, it holds $y_u + y_v \geq 1$.
    That is, $y$ is a fractional vertex cover.
    Specifically, for an edge $(u, v)$ where $u \in S_0$ and $v \in S_1$,
    using a similar argument to \cref{clm:tree-coupling},
    one can show that the algorithm does not distinguish between $u$ and $v$ with high probability.
    Therefore, we must have $\card{y_u - y_v} = o(1)$,
    and thus $y_u \geq \frac{1}{2} - o(1)$.
\end{proof}

\begin{lemma}
    Any \emph{deterministic} non-adaptive LCA that given a graph $G$ from the modified input distribution (as described earlier in the section), $O(1)$-approximates the minimum vertex cover with constant probability, requires $\kappa = \Delta^{\Omega(\log \Delta / \log \log \Delta)}$ queries.
\end{lemma}
\begin{proof}
    We take an arbitrary constant $0 < c < 1$,
    and prove the lower bound for any deterministic algorithm that $c$-approximates the minimum vertex in expectation.
    
    Observe that the expected size of the vertex cover outputted by the algorithm is $\sum_u y_u$.
    Due to \cref{clm:large-fractional-cover}, it holds
    $$
    \sum_u y_u \geq \card{S_0} \left(\frac{1}{2} - o(1)\right)
    = N_0 (1 - o(1)).
    $$
    On the other hand, the vertices $V(G) \setminus S_0$ form a vertex cover. Therefore, the size of the minimum vertex cover is at most
    $$
    \card{V(G) \setminus S_0} = 2 \epsilon N + 2 (N - N_0).
    $$
    Putting the two together, the approximation ratio is at most
    \begin{equation*}
    \frac{2 \epsilon N + 2 (N - N_0)}{N_0 (1 - o(1))}
    = (1 + o(1))\left(2\epsilon + \frac{r+1}{\delta-(r+1)} \right)
    \leq c. \qedhere     
    \end{equation*}
\end{proof}

By applying Yao's min-max principle, we obtain the final lower bound.
\begin{theorem}
    Any non-adaptive LCA that given a graph $G$, $O(1)$-approximates the minimum vertex cover with constant probability, requires $\Delta^{\Omega(\log \Delta / \log \log \Delta)}$ queries.
\end{theorem}

\section{Maximal Independent Set} \label{sec:mis}
In this section, we prove the same lower bound on query complexity for the maximal independent set problem. We require that the algorithm always outputs an independent set, which is maximal with constant probability.

\begin{theorem} \label{thm:mis-lb}
    Any non-adaptive LCA that, given a graph $G$, computes a maximal independent set with constant probability, requires $\Delta^{\Omega(\log \Delta / \log \log \Delta)}$ queries.
\end{theorem}

\begin{remark}
    The lower bound holds even for the almost-maximal independent set problem,
    i.e.\ where it is required that the output can be made maximal by adding at most $\epsilon n$ vertices.
\end{remark}

The argument follows closely the proof of \cref{thm:main}.
As such, we only describe the parts where they differ.

Similar to previous sections, we present an input distribution on which no deterministic non-adaptive LCA performs well (i.e.\ succeeds with constant probability). Then, the lower bound for (possibly) randomized non-adaptive LCAs follows from Yao's min-max principle.

To obtain the input distribution, we use the \emph{line graph} of the input for maximum matching approximation (\cref{sec:matching-lb}).
That is, we draw a graph $G = (V, E)$ from the input distribution for maximum matching.
Then, we construct a graph $G' = L(G)$ where each vertex corresponds to an edge in $E$, and two vertices are connected if the corresponding edges in $G$ share an endpoint.
Observe that any independent set in $G'$ corresponds to a matching in $G$.
In addition, a maximal independent set of $G'$ has the same size as the corresponding maximal matching in $G$, and hence provides a constant approximation of the maximum matching.

To prove the lower bound, it suffices to show that the vertices of $G'$ corresponding to significant ($C_0^\zero$--\,$C_0^\one$) and misleading ($C_0^\zero$--\,$C_1^\zero$) edges are indistinguishable.
This is argued by a step-by-step coupling of the query tree in the cases where the prompted edge is significant or misleading.
Therefore, we adapt the coupling of \cref{def:label-coupling} to fit the queries in the line graph.

The invariant we maintain through the coupling is that in any step $i$,
the edges $e_i$ and $f_i$ discovered in the two query trees have the same endpoint degrees.
That is, letting $e_i = (u_i, v_i)$ and $f_i = (x_i, y_i)$ where $v_i$ and $y_i$ are the endpoints previously not adjacent to any edges, it holds $\deg(u_i) = \deg(x_i)$ and $\deg(v_i) = \deg(y_i)$.
Then, assuming the invariant holds, the adjacent edges of $e_i$ and $f_i$ are coupled by simply using the label-based coupling of (vertex) neighbors (\cref{def:label-coupling}) on each endpoint.
More precisely, for each $e_{i+1}$ adjacent to $e_i$ through endpoint $u_i$, we use the label-based coupling of neighbors of $u_i$ and $x_i$ to obtain $f_{i+1}$ adjacent to $f_i$ through $x_i$.
Note that in this coupling, the edge $e_i$ is omitted from the adjacent edges of $u_i$, and $f_i$ is omitted from the adjacent edges of $x_i$.
Similarly, we couple the edges adjacent to $e_i$ and $f_i$ through endpoints $v_i$ and $y_i$ respectively.

Finally, it can be shown that this coupling maintains the invariant throughout the process with probability $1 - o(1)$, similar to \cref{clm:tree-coupling}. This implies that significant and misleading edges are indistinguishable, which in turn shows that the algorithm is unable to obtain a constant approximation.
\section{Applications of Non-Adaptive LCAs in MPC} \label{sec:mpc}
This section illustrates the application of non-adaptive LCAs to MPC algorithms.
We show that if for any small constant $\epsilon$, there exists a non-adaptive LCA for approximate minimum vertex cover with query complexity $\Delta^{(\log\Delta)^{1-\epsilon}}$,
then there is a sublinear-space MPC for $O(1)$-approximate maximum matching with round complexity $O\left((\log n)^{1/2 - \Omega(\epsilon)}\right)$,
which would be an improvement over the best known $\Ot(\sqrt{\log n})$-round algorithms of \cite{GhaffariU19,OnakArxiv18}. This approach is ruled out by our lower bound.

Throughout the section, we assume that the size of the maximum matching is at least linear in the number of vertices.
This is without loss of generality due to well-known \emph{vertex sparsification} techniques \cite{AssadiKL16,BehnezhadDH20}.
\begin{assumption}
    $\mu(G) = \Omega(n)$.
\end{assumption}

Here, we define the MPC model.
In the \emph{Massively Parallel Computation (MPC)} model,
an input of size $N$ is distributed among $M$ machines of local memory $S$.
During each round of computation, the machines can make arbitrary calculations on their local data, and then send messages to a number of other machines.
The only restriction is that the total size of the messages sent or received by a single machine cannot exceed $S$.
The efficiency of the algorithm is measured in the number of rounds.

For graph problems, the input size is $N = n + m$, and a popular choice of $S = n^\delta$ for a parameter $\delta \in (0, 1)$, i.e.\ the local memory of each machine is sublinear in the number of vertices.
This is known as the \emph{sublinear regime} (also \emph{fully scalable regime}), and is highly desirable as it is most suited for rapidly growing input sizes.

The main result of this section is formalized below.

\begin{theorem}\label{thm:MPC}
    For any $\epsilon > 0$, if there exists a non-adaptive LCA with query complexity \linebreak ${Q(\Delta) = \Delta^{(\log\Delta)^{1-\epsilon}}}$ and depth $D = \poly(\log n)$ that $O(1)$-approximates minimum vertex cover,
    then, for any constant parameter $0 < \delta < 1$, there exists an MPC algorithm that computes a constant approximation of the maximum matching in $(\log n)^{1/2 - \Omega(\epsilon)}$ rounds,
    using $O(n^\delta)$ local memory per machine and $O(nQ^2 + m)$ total space.
\end{theorem}

\begin{remark}
In proving \cref{thm:MPC}, we assume that the LCA also provides a matching that has size within a constant of its output vertex cover. This matching serves as the certificate that the vertex cover is indeed constant approximate. To our knowledge all existing MVC LCAs satisfy this:  they either come with a fractional matching that can be easily rounded or directly construct a maximal matching and pick the endpoints of its edges as the vertex cover.
\end{remark}

First, we describe the MPC algorithm on a high level (see \Cref{alg:mpc} for the pseudo-code). Then, we prove the correctness (\cref{clm:mpc-apx}), provide implementation details, and analyze the round complexity (\cref{clm:mpc-implementation}).

In each iteration, the algorithm computes a vertex set $A$ and a matching $M$ which matches a constant fraction of $A$,
then it deletes $A \cup V(M)$ along with some extra vertices that are high-degree, making the graph significantly sparser.
Once the graph is sparse enough, say $\Delta = \poly(\log n)$, the algorithm computes a maximal matching using known algorithms for low-degree graphs in $O(\log \Delta + \poly (\log \log n))$ rounds \cite{BarenboimEPS12}.
The algorithm outputs the union of these matchings, which match a constant fraction of the vertices in $G$ if we ignore the extra deleted vertices.
Therefore, by upper-bounding the number of the extra vertices, we can prove the algorithm obtains a constant approximation.

To be more precise, in each iteration, a subgraph of $H \subseteq G$ is obtained by sampling each edge of $G$ independently with probability $p$.
Then, the LCA is invoked for every vertex in $H$ to obtain a vertex cover $A$ and a matching $M \subseteq E(H)$ which matches a constant fraction of $A$. 
The probability $p$ is chosen small enough such that the query tree of the LCA for a single vertex fits into one machine and can be collected efficiently, i.e.\ $Q(p\Delta)^2 = O(n^\delta)$ where $p\Delta$ is (roughly) the maximum degree of $H$.
Finally, the algorithm deletes all the vertices of $A$, and afterwards deletes any vertex of degree more than $(\log n) / p$, hence significantly decreasing the maximum degree.
It can be shown that deleting the vertex cover of $H$ leaves at most $n/p$ edges.
Therefore, the number of extra deleted vertices in each iteration is bounded by $O(n/\log n)$, which is fine since the number of iterations is $o(\log n)$.

 \begin{algorithm}
     \caption{An MPC algorithm for constant-approximation of maximum matching, given a non-adaptive LCA for minimum vertex cover}
     \label{alg:mpc}
     \textbf{Input:} Graph $G$ with $n$ vertices
     
     Let $G_1 \gets G$, and $i \gets 1$

     \While{$G_i$ is non-empty \label{step:while}}{
        $\Delta_i \gets $ maximum degree in $G_i$

        \uIf{$\Delta_i \leq \log^2 n$ \label{step:terminating-condition}}{
            Compute a maximal matching $M_i$ of $G_i$ in $O(\log \Delta_i + \poly(\log \log n))$ rounds
            
            Let $A_i \gets V(M_i)$

            \Return $\bigcup_{j\leq i} M_j$
        }
        \Else{
            Let $p_i$ be such that $Q(10p_i\Delta_i)^2 \leq n^\delta$, i.e.\  $p_i  \gets \frac{2^{((\delta/2) \log n)^{1/(2 - \epsilon)}}}{10 \Delta_i}$

            Let $H_i$ be a subgraph of $G_i$, where each edge is sampled with probability $p_i$

            Run the non-adaptive LCA from every edge to obtain a vertex cover $A_i$ of $H_i$ and a matching $M_i \subseteq E(H_i)$. \label{step:lca}

            $G_i' \gets G_i - (A_i \cup V(M_i))$
            
            Let $D_i \gets \{ u \in V(G_i') \mid \deg_{G_i'}(u) > (\log n)/p\}$

            $G_{i+1} \gets G_i' - D_i$ \label{step:deleting-high-degree}

            $i \gets i + 1$
        }
     }
 \end{algorithm}

We let $G_i$ be the graph left at the beginning of the $i$-th iteration,
$\Delta_i$ be the maximum degree of $G_i$,
$p_i$ be the sampling probability,
$H_i$ be the sampled subgraph,
$A_i$ be the vertex cover of $H_i$, and
$M_i$ be the computed matching.
First, we upper-bound the number of iterations.

\begin{claim} \label{clm:number-of-iterations}
    There are $O\left( (\log n)^{1/2 - \Omega(\epsilon)} \right)$ iterations of the while-loop (Step~\ref{step:while}).
\end{claim}
\begin{proof}
    Note the choice of the sampling probability in iteration $i$:
    $$
    p_i  = \frac{2^{((\delta/2) \log n)^{1/(2 - \epsilon)}}}{10 \Delta_i}
    = \frac{2^{(\log n)^{1/2 + O(\epsilon)}}}{\Delta_i}.
    $$
    Since we delete the vertices of degree more than $(\log n)/p$ (in Step~ \ref{step:deleting-high-degree}),
    it holds
    $$
    \Delta_{i+1} \leq \frac{\log n}{p_i} \leq \frac{(\log n)\Delta_i}{2^{(\log n)^{1/2 + O(\epsilon)}}} \leq  \frac{\Delta_i}{2^{(\log n)^{1/2 + O(\epsilon)}}}.
    $$
    That is, the maximum degree decreases by a factor of $2^{(\log n)^{1/2 + O(\epsilon)}}$ every iteration.
    Therefore, by induction we get
    $$
    \Delta^{i+1} \leq \frac{\Delta_1}{2^{i\cdot (\log n)^{1/2 + O(\epsilon)}}}
    \leq \frac{n}{2^{i\cdot (\log n)^{1/2 + O(\epsilon)}}}
    = 2^{\log n - i\cdot (\log n)^{1/2 + O(\epsilon)}}.
    $$
    As a result, after at most $\displaystyle \frac{\log n}{(\log n)^{1/2 + O(\epsilon)}} = (\log n)^{1/2 - \Omega(\epsilon)}$ iterations,
    the maximum degree becomes small enough for the terminating condition to hold (Step~\ref{step:terminating-condition}).
\end{proof}

Now, we present a useful lemma that characterizes the number of remaining edges after removing the vertex cover of the sampled subgraph.

\begin{lemma} \label{clm:average-degree-after-maximal}
    Given a graph $G= (V, E)$, and $p \in [0, 1]$, let $H \subseteq G$ be a subgraph such that each edge is included independently with probability $p$,
    and $A$ be a vertex cover of $H$.
    Then, $\card{E[V(G) \setminus A]}$, the number of edges not adjacent to $A$, is at most $n/p$ with high probability.
\end{lemma}
\begin{proof}
    Fix a vertex set $B \subseteq V$ with $\card{E[B]} \geq n/p$.
    We upper-bound the probability of $V(G) \setminus A$ being equal to $B$ by $e^{-n}$, and take the union bound over all possible $B$.
    Here, the probability is over the randomness of $H$.
    Since $A$ is a vertex cover of $H$,
    the vertices $V(G) \setminus A$ form an independent set in $H$.
    Therefore, for $B = V(G) \setminus A$ to happen,
    no edge inside $B$ should be sampled.
    Thus,
    \begin{align*}
        \Pr(V(G)\setminus A = B)
        &\leq \Pr(E[B] \cap H = \emptyset) \\
        &= (1 - p)^{\card{E[B]}} \\
        &\leq (1-p)^{n/p} \\
        &\leq e^{-n}.
    \end{align*}
    Taking the union bound over all vertex sets $B$ with $E[B] \geq n/p$, we get
    \begin{equation*}
    \Pr(\card{E[V(G) \setminus A]} \geq n/p)
    \leq 2^n \cdot e^{-n} 
    = \left(\frac{2}{e}\right)^n. \qedhere
    \end{equation*}
\end{proof}

The following claim is essential to the analysis of the approximation ratio.

\begin{claim} \label{clm:number-of-deleted-vertices}
    The total number of extra vertices removed across iterations (in Step~\ref{step:deleting-high-degree}) is $o(n)$, with high probability.
\end{claim}
\begin{proof}
    We prove that the number of extra vertices removed in each iteration is at most $\frac{2n}{\log n}$. Then the claim follows since the number of iterations is at $o(\log n)$ by \cref{clm:number-of-iterations}.
    The number of extra removed vertices can be bounded by an averaging 
    argument.
    With high probability, there are at most $n/p_i$ edges of $G_i$ remaining we remove the vertices $A_i \cup V(M_i)$.
    Therefore, there at most $\frac{2n}{\log n}$ vertices have more than $(\log n)/p_i$ edges.
\end{proof}

\begin{claim} \label{clm:mpc-apx}
    \Cref{alg:mpc} computes a constant approximation of the maximum matching, with high probability.
\end{claim}
\begin{proof}
    Assume that each matching $M_i$ covers a $c$-fraction of the vertices in $A_i$, and thus
    $$
    \card{M_i} \geq \frac{c}{2}\card{A_i}.
    $$
    On the other hand, since $A_i$ is a vertex cover for $H_i$, and $M_i$ is a matching in it, we have
    $$
    \card{M_i} \leq \card{A_i}.
    $$
    Putting the two together, we can conclude:
    $$
    \frac{\card{M_i}}{\card{A_i \cup V(M_i)}}
    \geq \frac{\card{M_i}}{\card{A_i} + 2\card{M_i}}
    \geq \frac{(c/2) \card{A_i}}{\card{A_i} + 2\card{A_i}}
    \geq \frac{c}{6}.
    $$
    
    Now, let $A := \bigcup_{j \leq i} M_j$, $M := \bigcup_{j \leq i} M_j$, and $D := \bigcup_{j \leq i} D_j$.
    Observe that 
    $$V(G) \setminus D = \bigcup_{j\leq i}(A_j \cup V(M_j)).$$
    Therefore, the analysis in the previous paragraph implies
    $$
    \card{M} \geq \frac{c}{6}\card{V(G) \setminus D} \geq \frac{c}{3}\mu(G\setminus D).
    $$
    In addition, by \cref{clm:number-of-deleted-vertices}, we get
    $$
    \mu(G\setminus D) \geq \mu(G) - \card{D} = \mu(G) - o(n) = (1 - o(1))\mu(G).
    $$
    Combining the two, we can conclude that the approximation ratio is $\frac{c}{3} - o(1)$, which concludes the proof.
\end{proof}

\begin{claim} \label{clm:mpc-implementation}
    \Cref{alg:mpc} can be implemented in $(\log n)^{1/2 - \Omega(\epsilon)}$ rounds of the MPC model, using $S = n^\delta$ local space per machine, and $O(nQ^2 + m)$ total space, with high probability.
\end{claim}
\begin{proof}
    The only non-trivial part is Step~\ref{step:lca}.
    We show this step can be implemented in $O(\log D) = O(\log \log n)$, where recall $D = \poly(\log n)$ is the depth of the LCA.
    Then, since there are \linebreak $(\log n)^{1/2 - \Omega(\epsilon)}$ iterations, the total number of MPC rounds is 
    $$
    (\log n)^{1/2 - \Omega(\epsilon)} \cdot O(\log D) = (\log n)^{1/2 - \Omega(\epsilon)}.
    $$

    We shall prove that using $O(\log D)$ rounds, one can collect the explored subgraph of the non-adaptive LCA starting at each vertex of the graph. Then, the LCA can be simulated in one round.
    By the Chernoff bound, in iteration $i$, the maximum degree in the sampled subgraph $H_i$ is at most $10p_i\Delta_i$ w.h.p. 
    In addition, the choice of sampling probability 
    $p_i  = \frac{2^{((\delta/2) \log n)^{1/(2 - \epsilon)}}}{10 \Delta_i}$
    is such that $Q(10p_i\Delta_i)^2 \leq n^\delta$.

    On a high level, we collect the query tree in $\log D$ phases, doubling the depth of the collected tree in every phase.
    Let the non-adaptive LCA be represented by a rooted tree $T$.
    For each vertex $u$, we use $T_u^d$ to denote the subtree of $u$, truncated after distance $d$.
    That is, $T_u^d$ is obtained by taking the subtree of $u$, and deleting every descendant more than $d$ edges away from $u$.
    We designate a machine $M_v$ for each vertex $v$ in $G$.
    Then, in phase $j$, for all possible $u \in T$, we collect $T_u^{2^j}$ using $v$ as a starting point (denoted by $T_u^{2^j}(v)$), in $M_v$.

    The base case is $j=0$.
    Fix a vertex $v \in V(G)$.
    For each vertex $u \in T$, the tree $T_u^1$ consists only of the immediate children of $u$.
    We collect $T_u^1(v)$ in $M_v$ by simply collecting the corresponding neighbors of $v$.

    For $j > 0$, we shall make use of the previously collected trees of depth $2^{j-1}$.
    Fix a vertex $v \in V(G)$.
    For each vertex $u \in T$, we collect $T_u^{2^j}(v)$ as follows.
    Consider $T_u^{2^{j-1}}(v)$;
    The goal is to append the truncated parts with the appropriate subtrees.
    For each leaf $u'$ of $T_u^{2^{j-1}}$ that is at distance exactly $2^{j-1}$ from $u$, we look at the corresponding vertex $v' \in T_u^{2^{j-1}}(v)$, and send a message to $M'_v$ to collect $T_{u'}^{2^{j-1}}(v')$.
    Appending all the collected trees to $T_u^{2^{j-1}}(v)$,
    we obtain $T_u^{2^{j}}(v)$.

    For the last part of the argument, we use $Q$ as a shorthand for $Q(10p_i\Delta_i)$.
    As the final step, we show that the above operations can be done using $O(Q^2)$ local space per machine.
    To see this, note that each machine $M_v$, for each $u \in T$, is sending messages to at most $Q$ machines $M_v'$ to request their subtree ($Q^2$ requests in total).
    For each $u$, the sum of the sizes of the returned subtrees will not exceed $Q$ ($Q^2$ total size of responses).
    To answer these requests, each machine may need to respond to many more than $Q$.
    However, this can be handled using standard techniques such as sorting and broadcast trees since the total size of the messages is bounded by $MQ^2$ as previously described.
\end{proof}

\bibliographystyle{plainnat}

\bibliography{references}

\appendix
\section{Omitted Proofs of \cref{subsec:cluster-trees}} 
\label{apx:cluster-trees}

This section includes the omitted proofs of \cref{subsec:cluster-trees}.
The statements are repeated here for the reader's convenience.

The tree is defined recursively by a parameter $r$, and denoted with $\mT_r = (\mC_r, \mE_r)$. Here, $\mC_r$ is the set of clusters and $\mE_r$ is the set of edges.
For two adjacent clusters $X$ and $Y$, the edge $(X, Y) \in \mE_r$ is labeled with an integer $d(X, Y)$ which represents the degree of the vertices in $X$, from the edges going to $Y$ in the final construction ($d(Y, X)$ is defined similarly). We refer to $d(X, Y)$ as the degree of $X$ to $Y$.
The labels depend on a parameter $\delta> 0$ that is defined later.

\begin{definition}[The cluster tree \cite{KuhnMW16}]
    \label{mydef:cluster-tree}
    The base case is $r = 0$: \footnote{We use a slightly different base case than \cite{KuhnMW16}, which ultimately leads to the same tree structure.}
    \begin{align*}
        \mC_0 &:= \{C_0, C_1\}, \\
        \mE_0 &:= \{(C_0, C_1)\}, \\
        d(C_0, C_1) &:= 1, \qquad d(C_1, C_0) := \delta,
    \end{align*}
    and the tree is rooted at $C_0$.
    
    For $r > 0$, the tree $\mT_r$ is constructed by adding some leaf clusters to $\mT_{r-1}$:
    \begin{enumerate}
        \item For every non-leaf cluster $C \in \mC_{r-1}$, a child cluster $C'$ is created with $d(C, C') = \delta^r$.
        \item For every leaf cluster $C \in \mC_{r-1}$, let $p_C$ be the parent of $C$ and $i^*$ be the integer such that $d(C, p_C) = \delta^{i^*}$.
        Then, for every $i \in \{0, 1, \ldots, r\} - \{i^*\}$
        a child cluster $C'$ is added with $d(C, C') = \delta^i$.
    \end{enumerate}
    For every new edge $(C, C')$ where $C'$ is the leaf, we let the upward label $d(C', C) = \delta \cdot d(C, C')$.
    For every cluster $C$, we use $d_p(C)$ as a shorthand for $d(C, p_C)$.
\end{definition}

The \emph{degree} and \emph{color} of the clusters are defined as follows.

\begin{definition}[Cluster degrees]
    For a cluster $C \in \mC_r$, we define its degree in the tree as 
    $$
    d_\mT(C) := \sum_{C': (C, C') \in \mE_r} d(C, C'),
    $$
    We also use $\Delta_r$ to denote the maximum degree among the clusters.
\end{definition}

\begin{definition}[Color] \label{mydef:color}
    For a cluster $C \in \mC_r$, its color $\tau(C)$ is the integer $0 \leq i \leq r$ such that $C$ was first added in $\mT_i$.
\end{definition}

First, we characterize the degrees.

\begin{claim} \label{myclm:max-parent-degree}
    For all clusters $C \in \mC_r \setminus \{C_0\}$, it holds
    $$
    d_p(C) \leq \delta^{\tau(C) + 1}.
    $$
\end{claim}
\ifbool{showstuff}{ \begin{proof}
    Recall that $C$ is added in iteration $\tau(C)$, at which point the degree to any leaf cluster is at most $\delta^{\tau(C)}$. Therefore, we have
    \begin{equation*}
    d_p(C) = d(C, p_C) = \delta \cdot d(p_C, C) \leq \delta^{\tau(C) + 1}. \qedhere
    \end{equation*}
\end{proof}}{}

\begin{claim} \label{myclm:non-leaf-degrees-in-tree}
    For all non-leaf clusters $C \in \mC_r$, it holds
    $$
    d_\mT(C) = \sum_{i = 0}^r \delta^i =: \bar{d}_r.
    $$
\end{claim}
\ifbool{showstuff}{ \begin{proof}
    Take a non-leaf cluster $C$, and consider the iterative construction of $\mT_r$. Let $i^*$ be such that $d_p(C) = \delta^{i^*}$.
    Recall that $C$ is added in iteration $\tau(C) < r$ at which point it is a leaf.
    In iteration $\tau(C) + 1$, for every $i \in \{0, 1, \ldots, r\} - \{i^*\}$ a leaf $C'$ is added to $C$ with $d(C, C') = \delta^i$.
    Considering the degree to the parent $d_p(C) = \delta^{i^*}$, the total degree of $C$ at this point is
    $$
    \sum_{i=0}^{\tau(C) + 1} \delta^i.
    $$
    This proves the claim for cases where $r = \tau(C) + 1$.
    Otherwise, in the proceeding iterations $i$ from $\tau(C) + 2$ to $r$,
    a leaf $C'$ is added to $C$ with $d(C, C') = \delta^i$.
    Therefore, after the $r$-th iteration it holds:
    \begin{equation*}
    d_\mT(C) = \sum_{i = 0}^r \delta^i. \qedhere
    \end{equation*}
\end{proof}}{}

\begin{claim} \label{myclm:tree-max-degree}
    The maximum degree among the clusters is $\Delta_r = \delta^{r + 1}$.
\end{claim}
\ifbool{showstuff}{ \begin{proof}
    The maximum degree among the leaf clusters is $\delta^{r+1}$,
    because for a leaf cluster $C$, it holds
    $$
    d_\mT(C) = d_p(C) = \delta \cdot d(p_C, C),
    $$
    and $d(p_C, C)$ ranges from $\delta^0$ to $\delta^r$.
    For any non-leaf cluster, by \cref{myclm:non-leaf-degrees-in-tree}, the degree is 
    \begin{equation*}
    \sum_{i = 0}^r \delta^i
    = \frac{\delta^{r+1} - 1}{\delta - 1}
    \leq \delta^{r+1}.
    \qedhere
    \end{equation*}
\end{proof}}{}

The following claim states that clusters with the same total degrees
have similar outgoing labels.
This is a crucial fact in our coupling arguments.

\begin{claim} \label{myclm:same-outgoing-labels-in-tree}
    Take any two clusters $C, C' \in \mC_r$.
    If $d_\mT(C) = d_\mT(C')$, then $C$ and $C'$ have the same set of outgoing edge-labels.
\end{claim}
\ifbool{showstuff}{ \begin{proof}
    If $d_\mT(C) = d_\mT(C') = \bar{d}_r$, then both $C$ and $C'$ are non-leaf clusters.
    Therefore, for every $i \in \{0, 1, \ldots, r\}$ they have exactly one edge with outgoing label $\delta^i$.
    If $d_\mT(C) = d_\mT(C') \neq \bar{d}_r$, then $C$ and $C'$ are both leaves, and they only have one edge with outgoing label $d_p(C) = d_p(C')$.
\end{proof}}{}

The following claim characterizes the color of a child cluster based on the color of the parent, and the label of the edge coming in from the parent.

\begin{claim} \label{myclm:child-color}
    Take a cluster $C$ and let $C_i$ be a child of $C$ such that $d(C, C_i) = \delta^i$. Then, it holds
    $$
    \tau(C_i) = \max(i, \tau(C) + 1).
    $$
\end{claim}
\ifbool{showstuff}{ \begin{proof}
    Consider the iterative construction in \cref{mydef:cluster-tree}.
    Cluster $C$ is added in iteration $\tau(C)$ as a leaf.
    Let $i^*$ be such that $d_P(C) = \delta^{i^*}$.
    In iteration $\tau(C) + 1$, since $C$ is a leaf,
    for every $i \in \{0, 1, \ldots, \tau(C) + 1\} - \{i^*\}$,
    a child $C_i$ is added with $d(C, C_i) = \delta^i$.
    This proves the claim for $i \leq \tau(C) + 1$.

    In iteration $i > \tau(C) + 1$,
    since $C$ is no longer a leaf,
    a single child $C_i$ is added for $C$ with $d(C, C_i) = \delta^i$.
    This proves the claim for $i > \tau(C) + 1$.
\end{proof}}{}

The following claim states that the structure of the subtree of $C \notin \{C_0, C_1\}$ is uniquely determined by $\tau(C)$ and $d_P(C)$.
This observation is implicitly present in Definition 2 and Lemma 7 of \cite{KuhnMW16}.

\begin{claim} \label{myclm:identical-subtrees}
    Let $C, C' \in \mC_r \setminus \{C_0, C_1\}$ be two clusters of the same color
    $\tau = \tau(C) = \tau(C')$ 
    that have the same degrees 
    to their parents.
    Then, the subtrees of $C$ and $C'$ are identical (considering the labels on the edges).
\end{claim}
\ifbool{showstuff}{ \begin{proof}
For a fixed $r$, we prove the claim by induction on $\tau$.
The base case is $\tau = r$ where the claim holds because both $C$ and $C'$ are leaves.

Assume the claim holds for colors larger than $\tau$,
and let $C$ and $C'$ be two clusters such that 
$$\tau(C) = \tau(C') = \tau \qquad \textnormal{and} \qquad d_p(C) = d_p(C') = \delta^{i^*}.$$
For every $i \in \{0, 1, \ldots, r\} - \{i^*\}$, cluster $C$ has a child cluster $C_i$ with $d(C, C_i) = \delta^i$. Similarly, $C'$ has a child $C'_i$ with $d(C', C'_i) = \delta^i$.
It holds $$d_p(C_i) = d_p(C'_i) = \delta^{i+1}.$$
Furthermore, by \cref{myclm:child-color}, we have
$$
\tau(C_i) = \tau(C'_i) \geq \tau + 1.
$$
Therefore, by the induction hypothesis, the subtree of $C_i$ is identical to that of $C'_i$. This implies the claim for the subtrees of $C$ and $C'$ and concludes the proof.
\end{proof}}{}

The following claim uses \cref{myclm:identical-subtrees}
to show that starting from two clusters $C$ and $C'$, taking the edges with the same labels leads to isomorphic subtrees if the edge-label is large enough.

\begin{claim} \label{myclm:large-steps}
    Take two clusters $C, C' \in \mC_r$. Let $\tau = \max(\tau(C), \tau(C')) + 2$, and $B$ and $B'$ be \emph{neighbors} of $C$ and $C'$ resp.\ such that $d(C', B) = d(C, B') \geq \delta^{\tau}$. Then, the subtrees of $B$ and $B'$ are identical (considering the labels on the edges).
\end{claim}
\ifbool{showstuff}{ \begin{proof}
    Recall that $d_p(C) \leq \delta^{\tau(C) + 1} < \delta^\tau$ (\cref{myclm:max-parent-degree}). The same holds for $C'$. 
    Therefore, $B$ and $B'$ must be children of $C$ and $C'$ respectively.
    Let $i \geq \tau$ be such that $d(C, B) = d(C', B') = \delta^i$. Then, by \cref{myclm:child-color}, we have 
    $$
    \tau(B) = \max( i, \tau(C) + 1) \qquad \textnormal{and} \qquad
    \tau(B') = \max( i, \tau(C') + 1).
    $$
    Since $i \geq \tau = \max(\tau(C), \tau(C')) + 2$, it holds $\tau(B) = \tau(B')$.
    Hence, we can invoke \cref{myclm:identical-subtrees} to derive the claim.
\end{proof}}{}

\end{document}